\def\arxiv{}
\ifx\arxiv\undefied % JMP header

\documentclass[jmp]{revtex4-1}

\else % arXiv header

\documentclass[12pt]{article}

\usepackage[verbose,tmargin=1in,bmargin=1in,lmargin=1in,rmargin=1in]{geometry}

\fi

\usepackage[utf8]{inputenc}

\usepackage[unicode=true,pdfusetitle,
 bookmarks=true,bookmarksnumbered=false,bookmarksopen=false,
 breaklinks=true,pdfborder={0 0 0},backref=false,colorlinks=true]{hyperref}
\hypersetup{linkcolor=blue,citecolor=blue,urlcolor=blue}

\usepackage{amsmath,amsthm,amsfonts}
\usepackage{xrf}
\usepackage{enumerate}
\usepackage{slashed}
\usepackage{xcolor}
\usepackage{mathabx}

\usepackage[normalem]{ulem}

\newcommand{\db}{}%\color{red}}
\newcommand{\ed}{}%\color{black}}
\newcommand{\old}[1]{}%\ifx\showold\undefined\else\color{gray}\ \sout{#1}\ \ed\fi}
\newcommand{\cross}[1]{}%{\xout{#1}}}

\def\R{{\mathbb R}}  %%
\def\N{{\mathbb N}}  %%
\def\C{{\mathbb C}}  %%
\def\F{{\mathcal{F}}}

\renewcommand{\cal}[1]{{\mathcal #1}}
\renewcommand{\vec}[1]{{\mathbf #1}}

\newcommand{\im}{\operatorname{Im}}
\newcommand{\re}{\operatorname{Re}}

\newcommand{\sk}[1]{\left\langle #1\right\rangle}

\newcommand{\causal}{\operatorname{Causal}}

\newtheorem{fact}{Fact}[section]

\newtheorem{theorem}[fact]{Theorem}

\newtheorem{corollary}[fact]{Corollary}

\newtheorem{definition}[fact]{Definition}

\newtheorem{lemma}[fact]{Lemma}

\newtheorem{remark}[fact]{Remark}

\renewcommand{\vec}[1]{\mathbf{#1}}

\newcommand{\supp}{\operatorname{supp}}
\newcommand{\selfmaps}{\righttoleftarrow}

\newcommand{\sol}{0}
\newcommand{\M}{\mathcal{M}}
\newcommand{\cH}{\mathcal{H}}
\newcommand{\cC}{\mathcal{C}}

\newcommand{\CThree}{{\cC_{3}}}
\newcommand{\CM}{\cC_{\M}}
\newcommand{\CSigma}{{\cC_\Sigma}}
\newcommand{\Cs}{{\cC_\sol}}
\newcommand{\CA}{{\cC_A}}

\newcommand{\HThree}{{\cH_3}}
\newcommand{\HM}{\cH_\M}
\newcommand{\HSigma}{\cH_\Sigma}
\newcommand{\Hs}{\cH_\sol}
\newcommand{\HA}{\cH_A}

\begin{document}

\title{\textsc{Dirac Equation with External Potential and Initial Data on Cauchy
Surfaces}}

\ifx\arxiv\undefined % JMP

\author{D.-A. Deckert}
\email{deckert@math.ucdavis.edu}
\affiliation{%
  Mathematical Department of the University of California Davis,
  One Shield Ave, CA 95616 Davis,
  USA
}

\author{F. Merkl}
\email{merkl@mathematik.uni-muenchen.de}
\affiliation{%
  Mathematisches Institut der Ludwig-Maximilians-Universit\"at M\"unchen,
  Theresienstr. 39, 80333 M\"unchen, Germany
}

\else % arXiv

\author{D.-A. Deckert\\
  \small Mathematical Department of the University of California Davis\\
  \small One Shield Ave, CA 95616 Davis, USA\\
  \small deckert@math.ucdavis.edu\\
  \\
  F. Merkl\\
  \small Mathematisches Institut der Ludwig-Maximilians-Universit\"at M\"unchen\\
  \small Theresienstr. 39, 80333 M\"unchen, Germany\\
  \small merkl@mathematik.uni-muenchen.de
}

\maketitle

\fi

\begin{abstract}
    With this paper we provide a mathematical review on the initial-value problem of the one-particle Dirac equation on
    space-like Cauchy hypersurfaces for compactly supported external potentials.
    We, first, discuss the physically relevant spaces of solutions and initial values
    in position and mass shell representation; second, review the action of the
    Poincaré group as well as gauge transformations on those spaces; third, introduce
    generalized Fourier transforms between those spaces and prove
    convenient Paley-Wiener- and Sobolev-type estimates. These generalized Fourier
    transforms immediately allow the construction of a unitary evolution
    operator for the free Dirac equation
    between the Hilbert spaces of square-integrable
    wave functions of two respective Cauchy surfaces. With a
    Picard-Lindel\"of argument this evolution map is generalized to the Dirac
    evolution including the external potential. For the latter we introduce a
    convenient interaction picture on Cauchy surfaces. 
    These tools immediately provide another proof of the well-known existence
    and uniqueness of classical solutions and their causal structure. 
\end{abstract}

\ifx\arxiv\undefined % JMP

\maketitle

\fi

\tableofcontents

\section{Introduction and Motivation}

The one-particle Dirac equation plays a fundamental role in relativistic quantum
theory. It was \old{invented}\db introduced \ed by Dirac to describe the
dynamics of spin 1/2 fermions such as electrons. Although as a one-particle
equation alone its physical interpretation is difficult due to the occurrence of
negative energy states, its second-quantized form leads to the so-called
\emph{external field} or \emph{no-photon quantum electrodynamics}. \db In this
respect, \ed our interest in the solution theory for initial data on space-like
hypersurfaces is \db three-fold\ed:

\begin{enumerate}[(1)]
  
    \item  While \old{the}\db a \ed mathematical rigorous construction of the
        second-quantized time evolution in a fixed fermionic Fock space has only
        been carried out successfully in the case of zero space-like components
        of the external four-vector potential, \old{physicist nevertheless
        developed powerful techniques to extract predictions from it by
        means of formal manipulations of perturbation series}\db physicists have
        developed powerful recipes to extract predictions in terms of formal
        perturbation series from external field quantum electrodynamics despite
        the ill-defined nature of its equations of motion\db; e.g., see
        \cite{Dyson:06,Schweber:1961zz} for an overview\ed.  \old{These
        techniques work best in the static case, i.e., in the computation of
        the reaction of the quantum vacuum to static external fields, or in
        the scattering case, i.e., in the computation of transition
        probabilities between asymptotically free states when scattered at a
        prescribed potential.}Whether the resulting \old{formal}series do
        converge or in which regimes the corresponding corrections are small
        seems to be unknown.  This fact becomes particular \old{problematic}\db
        unsatisfactory \ed in light of next generation laser experiments such as
        planned to be conducted, e.g., at the
        Extreme Light Infrastructure \cite{ELI}. These
        will allow to probe quantum electrodynamics in
        strong-field regimes in which, first, the conventional perturbative
        techniques become questionable, and second, a mere scattering theoretic
        description \db of physical processes is not sufficient \ed and a
        dynamical description is needed; cf.\ \cite{2009EPJD...55..327D}.

        The obstacle in the construction of \old{such}a second-quantized time
        evolution was observed in the works
        \cite{shale:65,Ruijsenaars:77-1,Ruijsenaars:77-2}.  There it is shown
        that the one-particle time evolution can \old{only}be lifted to the Fock
        space if and only if the space-like components of the external
        four-vector potential are zero. One way out of this dilemma, as
        sketched in \cite{fierz-scharf-79}, is to implement the time evolution
        on time-varying Fock spaces.  Two different such constructions have been
        carried out in \cite{langmann:96,mickelsson:98} and
        \cite{Deckert:2010ii}. Both \old{constructions}involve additional
        degrees of freedom \db(such as the charge-renormalization) \ed which can
        be encoded in \old{a}\db the \ed choice of a phase depending on the
        external field\db; see \cite{Deckert:2010ii} for the identification of
        the dependence of these degrees of freedom on the external field\ed.
        The resulting second-quantized time evolution transports initial data
        from one equal-time hyperplane to another and gives rise to unique
        transition probabilities.  However, quantities \db such \ed as the
        charge-current density depend manifestly on \old{such a}\db this
        unidentified \ed phase.  \db In particular, this concerns the so-called
        phenomenon of \emph{vacuum polarization} but also the dynamical
        description of pair creation processes for which so far only a few
        rigorous treatments are available; see \cite{Gravejat:2013kr} for vacuum
        polarization in the Hartree-Fock approximation for static external
        sources and \cite{PicklDuerr:08-1} for adiabatic pair creation\ed.  The
        \old{remaining}\db involved \ed degrees of freedom can be reduced
        further by imposing the Bogolyubov causality condition
        \cite[(17.30)]{andSHIRKOVDmitryVasilevich:1959wy} or more or less
        equivalently by implementing second-quantized evolution maps between
        Fock spaces associated to space-like Cauchy hypersurfaces\db, which is
        the content of a follow-up work\ed.  During our study of the latter
        approach a detailed knowledge of the one-particle Dirac equation for
        initial data on space-like hypersurfaces proved to be essential.
        Collecting this knowledge is our main motivation for writing this paper.
        \old{The construction of second-quantized evolution maps between Fock
        associated to space-like Cauchy hypersurfaces will be the content of a
        follow-up paper.}\db 
        
        Another possible way out of the mentioned dilemma that deserves mentioning lies
        in a reformulation of quantum electrodynamics in terms of the so-called
        \emph{fermionic projector} for which we refer the reader to
        \cite{Finster:2006ue}.\ed

\end{enumerate}

\noindent Apart from this we have \old{also a more general
aim}\db two more general interests\ed:

\begin{enumerate}[(1)]
    \setcounter{enumi}{1}

  \item As \db well as \ed classical electrodynamics also quantum
      electrodynamics is a manifestly Lorentz covariant theory. However, its
      Lorentz covariance is often obscured in the presentation of the theory
      when its fundamental equations of motion are formulated exclusively on
      equal-time hyperplanes.  Non-trivial Lorentz boosts, however, tilt any
      equal-time hyperplane in space-time and, therefore, Lorentz covariance
      \db of quantum electrodynamics \ed is only apparent in the momentum
      representation of the free theory or \db in \ed the asymptotic description
      of \db the corresponding \ed scattering theory\old{ only}.  In order to
      make \old{the}Lorentz and gauge covariance explicit in a space-time
      representation \old{we follow the spirit of Tomonaga and Schwinger who
      worked exclusively with initial data on space-like Cauchy
      hypersurfaces}\db we provide necessary mathematical results for the Dirac
      equation that allow to work exclusively with initial data on space-like Cauchy
      hypersurfaces in the spirit of Tomonaga and Schwinger 
      \cite{Tomonaga:1946zz,Schwinger:1948wr}. In this regard our \ed efforts are
      intended to contribute towards a mathematically rigorous understanding of
      their works.  
  
  \item \old{In particular}Furthermore, by deforming a Cauchy surface in a small
      neighborhood of a point, the dynamics can be studied locally\db; compare
      the differential formulation of the equations of motion (25) in
      \cite{Tomonaga:1946zz}\ed.  Providing mathematical tools for such a
      study is \old{our second main}\db our third \ed motivation for this paper.
      \old{Due to the causal structure of solutions of the Dirac equation, which
      is also a topic of this paper, only local information about the
      external potential is needed.}\db One important observation is that
      \ed due to the causal structure of solutions of the Dirac equation
      only local information about the external potential is needed.  In
      consequence, the behavior of the external potential and of the Cauchy
      surfaces near ``infinity'' is irrelevant for \old{our purposes}finite
      times. \old{For technical convenience we therefore restrict}\db This
      justifies the technical convenience of restricting our study to
      compactly supported external potentials.

\end{enumerate}

There are several treatises of the Dirac equation in the classical literature
which discuss the Dirac equation on equal-time hyperplanes, most
prominently \cite{thaller:92}. 
The initial value problems for such hyperbolic systems of differential equations
was already treated in, e.g., \cite{john:82,taylor:11}.
Wave equations on Lorentzian manifolds including
the Dirac equation have been studied, e.g., in \cite{dimock:82}, \cite{baer:2007}, 
\cite{ringstroem:09}, \cite{Finster:2012}, and \cite{Derezinski:13}.
In particular, their results ensure existence and uniqueness of solutions 
and identify their causal structure. The main contribution of our work is the
introduction of general Fourier transform and corresponding
Paley-Wiener techniques 
that can be exploited in
flat space-time
to study solutions to the one-particle Dirac equation on Cauchy surfaces. 
A byproduct of these generalized Fourier transform methods yields yet another proof for
the well-posedness of the initial value problem of Dirac equation on Cauchy
surfaces. Therefore we used this opportunity to compile our results 
in the form of a self-contained review given in Section~\ref{sec:review} that
ranges from general assertions about the initial value problem to
a detailed analysis of solutions.
To increase readability the more technical proofs are provided separately in 
Section~\ref{sec:proofs}.

\paragraph{Acknowledgment.} The authors cordially thank Wojciech Dybalski and
Felix Finster for
their helpful and detailed suggestions on the classical and contemporary literature of the initial value problem of
hyperbolic systems of differential equations. 

\paragraph{Notation.} Positive constants are denoted by $C_1, C_2, C_3,\ldots$
They keep their value throughout the whole article. Any fixed quantity
a constant depends on is displayed at least once when the constant is introduced.

\section{The One-Particle Dirac Equation}\label{sec:review}

The one-particle Dirac equation for an electron of mass $m>0$ is given by
\begin{align}\label{eq:dirac-equation}
  (i\slashed\partial-\slashed A) \psi=m\psi,
\end{align}
where the external potential
\begin{align}\label{def-A}
  A=(A_{\mu})_{\mu=0,1,2,3}\in\cC^{\infty}_{c}(\mathbb R^{4},\mathbb R^{4}),
\end{align}
is assumed to be smooth and compactly supported.
In our notation the elementary charge $e$ \db(having a negative sign in
the case of an electron) \ed is already included in $A$ \db and units are chosen
such that $\hbar=1$ and $c=1$\ed.  
Moreover, the elements of $\R^4$ are represented by 
$x=(x^0,x^1,x^2,x^3)=(x^0,\vec x)=x^\mu e_\mu$,
where $e_\mu$ denotes the canonical basis vectors in $\R^4$.
We endow $\R^4$ with the metric tensor
$g=(g_{\mu\nu})_{\mu,\nu=0,1,2,3}=\operatorname{diag}(1,-1,-1,-1)$.
Raising and lowering indices is done w.r.t.\ this metric tensor. We employ
Einstein's summation convention, Feynman's slash-notation 
$\slashed{\partial}=\gamma^\mu \partial_\mu$, $\slashed A
= \gamma^\mu A_\mu$, and use the standard
representation of the Dirac matrices $\gamma^\mu\in\C^{4\times 4}$ that fulfill
$\{\gamma^\mu,\gamma^\nu\}=2g^{\mu\nu}$.

Before touching the question about existence of solutions to the Dirac equation
(\ref{eq:dirac-equation}) in Section~\ref{sec:main-results} we introduce
and study the physically relevant
classes of solutions and initial data in space-time and energy-momentum
representation.

\subsection{Relevant Spaces in Space-Time Representation}
\label{sec:spaces-cauchy-sol}

We now define the classes of possible solutions
to the Dirac equation (\ref{eq:dirac-equation}) and initial data 
in space-time representation considered in this work.

\begin{definition}[Classical Solutions in Space-Time Representation]\label{def:CA}
  Let $\CA$ denote the space of all smooth solutions
$\psi\in C^\infty(\R^4,\C^4)$ of the 
Dirac equation~(\ref{eq:dirac-equation}) 
which have a spatially compact causal support in the following sense:
There is a compact set $K\subset\R^4$ such that
\begin{equation}
\operatorname{supp}\psi\subseteq K+\causal,
\end{equation}
where $\causal:=\{x\in\R^4|\;x_\mu x^\mu\ge 0\}$
denotes the set of all causal vectors. 
\end{definition}

One way to build a solution theory for the Dirac equation
(\ref{eq:dirac-equation}) is to
generate solutions in $\CA$ from initial data prescribed 
on \emph{Cauchy surfaces} which, for our 
purposes, are defined as follows:

\begin{definition}[Cauchy Surfaces]\label{def:cauchy-surface}
We define a Cauchy surface $\Sigma$ in $\R^4$ to be a smooth, 3-dimensional
submanifold of $\R^4$ that fulfills the following three conditions:
\begin{enumerate}
 \item Every \db inextensible, \ed two-sided, time- or light-like, continuous path
in $\R^4$ intersects $\Sigma$ in a unique point.
 \item For every $x\in\Sigma$, the tangential space $T_x\Sigma$ is
space-like.
 \item The tangential spaces to $\Sigma$ are bounded away from light-like
   directions in the following sense: The only light-like accumulation point
   of $\bigcup_{x\in\Sigma}T_{x}\Sigma$ is zero.
\end{enumerate} 
\end{definition}
The \old{main}differences compared to\db, e.g., \ed the definition given in 
\cite[Section 8.3]{Wald:1984un}, are the smoothness condition 
as well as (b) and (c). Condition (c) is not essential, 
but convenient to use as we are mainly interested in the local and causal 
properties
of the Dirac evolution.
\begin{remark}
%\ifx\arxiv\undefied 
%\else
The simplest Cauchy surface is the time-zero hyperplane
\begin{align}
  \Sigma^0:=\left\{x\in\mathbb R^{4}\;|\;x^{0}=0\right\}.
\end{align}
Note that by definition, light-like vectors are not allowed as 
tangential vectors to $\Sigma$. For example
\begin{equation}
\{(\arctan x^1,\vec x)|\;\vec x=(x^1,x^2,x^3)\in\R^3\}
\end{equation}
is not a Cauchy surface according to our definition, as condition (b) is
violated. Moreover,
\begin{equation}
\{(\sqrt{1+\vec x^{\,2}},\vec x)|\;\vec x=(x^1,x^2,x^3)\in\R^3\}
\end{equation}
is not a Cauchy surface either, as condition (a) is violated.
%\fi
In coordinates, every Cauchy surface $\Sigma$ can be parametrized as 
\begin{equation}
\label{eq:parametrize Sigma}
\Sigma=\{(t_\Sigma(\vec x),\vec x):\;\vec x\in\R^3\}
\end{equation}
with a smooth function $t_\Sigma:\R^3\to\R$ that fulfills
$|\nabla t_\Sigma(\vec x)|<1$ for every $\vec x\in \R^3$.
Note that this condition is necessary but not sufficient for $\Sigma$ 
satisfying conditions (a) and (b). 
This is illustrated by the second counterexample above.
Condition (c) guarantees that even $\sup_{\vec x\in\R^3}|\nabla t_\Sigma(\vec x)|<1$
holds.
\end{remark}

In order to define the spaces of initial data it will be convenient to introduce
the following notation. The standard volume form over $\R^4$ is denoted by
$d^4x=dx^0\,dx^1\,dx^2\,dx^3$; the product of forms is understood as wedge
product. The symbol $d^3x$
means the 3-form $d^3x=dx^1\,dx^2\,dx^3$ on $\R^4$.
Contraction of a form $\omega$ with a vector $v$ is denoted
by $i_v(\omega)$.
The notation $i_v(\omega)$ is also used for the spinor matrix
valued vector $\gamma=(\gamma^0,\gamma^1,\gamma^2,\gamma^3)=\gamma^\mu e_\mu$:
\begin{equation}
i_\gamma(d^4x)=\gamma^\mu i_{e_\mu}(d^4x).
\end{equation}
Furthermore, for a $4$-spinor $\psi\in\C^4$ (viewed as column vector),  
$\overline{\psi}$ stands for the row vector $\psi^*\gamma^0$,
where ${}^*$ denotes \db hermitian \ed conjugation.

\begin{definition}[Spaces of Initial Data in Space-Time Representation]
\label{def: CSigma}
\label{def:HSigma}
For any Cauchy surface $\Sigma$ we define the vector space 
\begin{align}\label{eq:CSigma_dense}
  \CSigma:=C^\infty_c(\Sigma,\C^4).
\end{align}
%of all smooth, compactly supported, 4-spinor valued functions on $\Sigma$.
%The space $\CSigma$ is dense in $\HSigma$ \db w.r.t.\ \ed the previously introduced scalar product
%\begin{equation}
%\tag{\ref{eq:scalar-product}}
%\sk{\phi,\psi}
%=\int_{\Sigma}\overline{\phi(x)}i_\gamma(d^4x)\psi(x).
%\end{equation}
% Given any compact set $K\subset\Sigma$, any $n\in\N$, and $\psi\in\cal
% C_\Sigma$ we define
% \begin{align}
%   \Vert \psi \Vert_{\SigmaK,n} := 
%   \sup_{\substack{\vec x\in\R^3:\\(t_\Sigma(x),\vec x)\in K}}
%   \max_{\beta:\, |\beta|\leq n}
%   \left| D^\beta_{\vec x} \psi( t_\Sigma(\vec x), \vec x) \right|
%   \qquad \text{for } \supp \psi \subseteq K
% \end{align}
% and $\Vert \psi \Vert_{\SigmaK,n} := \infty$ otherwise.
For a given Cauchy surface $\Sigma$, 
let $\HSigma=L^2(\Sigma,\C^4)$ denote the  vector space of
all 4-spinor valued measurable functions $\phi:\Sigma\to \C^4$ (modulo changes on null sets)
having a finite norm
$\|\phi\|=\sqrt{\sk{\phi,\phi}}<\infty$ w.r.t.\ the scalar
product 
\begin{equation}
\label{eq:scalar-product}
\sk{\phi,\psi}
=\int_{\Sigma}\overline{\phi(x)}i_\gamma(d^4x)\psi(x).
\end{equation}
\end{definition}

To see that the pairing in (\ref{eq:scalar-product}) is a scalar product 
on $\CSigma\subset \HSigma$, note that for all future-directed
time-like vectors $n$, the $4\times 4$ matrix
\begin{align}\gamma^0\slashed{n}=(\gamma^0\slashed{n})^*\end{align}
is positive definite. Furthermore, condition (c) of Definition \ref{def:cauchy-surface}
ensures for all Cauchy surfaces $\Sigma$
\begin{align}
  \sup_{x\in\Sigma}\left\Vert\left(\gamma^0\slashed n(x)\right)^{-1}
  \right\Vert < \infty, 
  \label{eq:condition-c}
\end{align}
where the symbol $\|{\cdot}\|$ denotes the appropriate matrix
norm. In the special case $\Sigma^0=\{x\in\R^4|\;x^0=0\}$ the scalar product on 
$\cC_{\Sigma^0}$
reduces to the standard one:
\begin{equation}
\sk{\phi,\psi}=
\int_{\Sigma^0}\phi(x)^*\psi(x)\,d^3x.
\end{equation}
Note that $\CSigma$ is dense in $\HSigma$ w.r.t.\ the scalar product
(\ref{eq:scalar-product}).

\begin{remark}
For $x\in\Sigma$, the restriction of the \db spinor matrix \ed valued 3-form
$i_\gamma(d^4x)$ to the tangential space $T_x\Sigma$ is given by
\begin{equation}
\label{eq: repr igammad4x}
 i_\gamma(d^4x)=
\slashed{n}(x)i_n(d^4x)=
\left(
\gamma^0+\sum_{\mu=1}^3 \gamma^\mu \frac{\partial t_\Sigma(\vec{x})}{\partial x^\mu}
\right)d^3x
\text{ on }(T_x\Sigma)^3,
\end{equation}
where $n$ denotes the future-directed unit normal vector field to $\Sigma$
in the Minkowski sense. In physics one often uses the notation $d\sigma(x)=
i_n(d^4x)$.
\end{remark}

Finally, the space of solutions $\CA$ can be extended to a Hilbert space:

\begin{definition}[Hilbert Space of Solutions in Space-Time Representation]
We endow $\CA$ with the scalar product
\begin{equation}\label{eq:scalar-product-sol}
\sk{\phi,\psi}=
\int_{\Sigma}\overline{\phi(x)}i_\gamma(d^4x)\psi(x),
\end{equation}
where $\Sigma$ denotes any Cauchy surface and define
\begin{align}
  \HA := \operatorname{completion}(\CA) \label{eq:Cs_dense},
\end{align}
which denotes the (abstract) completion of $\CA$ w.r.t.\ 
the norm $\|\psi\|=\sqrt{\sk{\psi,\psi}}$. 
In case the external potential $A$ is zero we will use the notation
 $\Cs= \CA|_{A=0}$.
\end{definition}

Note that the scalar product (\ref{eq:scalar-product-sol}) is well-defined
for $\phi,\psi\in\CA$ because the support
of the form $\overline{\phi(x)}i_\gamma(d^4x)\psi(x)$
intersects $\Sigma$ in a compact set, and because
the integral does not depend on the choice of the Cauchy surface $\Sigma$.
This follows from Stokes' theorem \db as \ed the 3-form
$\overline{\phi(x)}i_\gamma(d^4x)\psi(x)$ is closed:\ed
\begin{align}
&d[\overline{\phi(x)}i_\gamma(d^4x)\psi(x)]=
\partial_\mu(\overline{\phi(x)}\gamma^\mu\psi(x))\,d^4x
\nonumber
\\&=
(\partial_\mu\overline{\phi(x)})\gamma^\mu\psi(x)\,d^4x
+\overline{\phi(x)}\gamma^\mu\partial_\mu\psi(x)\,d^4x
\nonumber
\\&=
\overline{\slashed{\partial}\phi(x)}\psi(x)\,d^4x
+\overline{\phi(x)}\slashed{\partial}\psi(x)\,d^4x
\nonumber
\\&=
i\overline{(m+\slashed A)\phi(x)}\psi(x)\,d^4x
-i\overline{\phi(x)}(m+\slashed A)\psi(x)\,d^4x=0.
\label{eq:stokes-argument}
\end{align}

\subsection{Relevant Spaces in Energy-Momentum Representation}
\label{sec:spaces-momentum}

In the same way the momentum representation, i.e., the standard Fourier
transform, provides a fundamental tool for the study of the non-relativistic
Schr\"odinger equation on equal-time hyperplanes, the energy-momentum
representation of solutions in $\Hs=\cH_{A=0}$ and the corresponding
generalized Fourier transform will facilitate the study of the Dirac equation
on Cauchy surfaces.  For this purpose we introduce the mass shell
\begin{align}
    \M=\{p\in\R^4|\;p_\mu p^\mu=m^2\}.
\end{align}
The mass shell $\M$ has two connected components which are denoted by
\begin{equation}
\M_+=\{p\in\M|\;p^0>0\},\quad
\M_-=\{p\in\M|\;p^0<0\}.
\end{equation}
We endow $\M$ with the orientation that makes the projection
$\M\to\R^3$, $(p^0,\vec p)\mapsto \vec p$ positively oriented.
Restricted to $\M_\pm$, this projection has the inverses
\begin{equation}
\R^3\ni \vec p\mapsto p_\pm(\vec p)=(\pm E(\vec p),\vec p)\in\M_\pm,
\quad\text{where} 
\quad
E(\vec p):=\sqrt{\vec p^{\,2}+m^2}.
\end{equation}

The free Dirac equation in momentum representation reads $\slashed
p\psi=m\psi$ for $p\in\M$. The corresponding solution space is given by
\begin{equation}
{\cal D}_p=\{\psi \in\C^4|\;
\slashed{p}\psi=m\psi\}.
\end{equation} 
We therefore introduce the complex vector bundle of rank 2 over $\M$
\begin{equation}
{\cal D}:=\{(p,\psi)|\;p\in\M,\;\psi\in {\cal D}_p\}
\end{equation}
which we call the Dirac bundle.
For $p\in\M$, the orthogonal projection
from $\C^4$ onto ${\cal D}_p$ w.r.t.\ the standard scalar product
is given by the matrix
\begin{equation}
\label{eq:Pp}
P(p)=\frac{\slashed p+m}{2p^0}\gamma^0.
\end{equation}
For $\vec p\in\R^3$, the vector spaces ${\cal D}_{p_+(\vec p)}$
and ${\cal D}_{p_-(\vec p)}$ are orthogonal complements to each other \db 
so that
\begin{equation} 
\label{eq: P+P-}
P_+ (\vec p)+P_- (\vec p)=1\in\C^{4\times 4},
\end{equation}
where we used the short-hand notation
\begin{equation}
\label{eq: P+-}
P_\pm (\vec p)=P(p_\pm(\vec p)).
\end{equation} \ed

With the application of Paley-Wiener arguments in mind to study support
properties of functions, we also introduce a
complexified version $\M_\C$.  
We define
\begin{equation}
\M_{\C}=\{p\in\C^4|\;p_\mu p^\mu=m^2\}.
\end{equation}
Note that $\M_{\C}$ is a connected submanifold of $\C^4$ 
of complex dimension $3$.
We use the following notations. The standard volume
form over $\R^3$ is denoted by $d^3\vec p=dp^1\,dp^2\,dp^3$ and one
has
\begin{equation}
i_p(d^4p)=
p^0\,dp^1\,dp^2\,dp^3-p^1\,dp^0\,dp^2\,dp^3+p^2\,dp^0\,dp^1\,dp^3
-p^3\,dp^0\,dp^1\,dp^2.
\end{equation}
For $p\in\M$, the restriction of this form to the tangential space
$T_p\M$ is the Lorentz invariant volume form on the mass shell
\begin{equation}\label{eq:ipd4p}
i_p(d^4 p)=\frac{m^2}{p^0}dp^1\,dp^2\,dp^3=
\frac{m^2}{p^0}d^3p\text{ on } (T_p\M)^3.
\end{equation}
The
\emph{euclidean} norm of $p\in\C^{d}$ for $d\in\N$ is denoted by $|p|$.
We introduce the following spaces:
\begin{definition}[Solutions in Energy-Momentum Representation]
    Let $\HM=L^2(\M,{\cal D})$ denote the space of all
square integrable sections $\psi$ in the Dirac bundle.
This means that $\HM$ consists 
of all measurable functions $\psi:\M\to \C^4$ 
(modulo changes on null sets)
that fulfill almost everywhere
\begin{align}\slashed{p}\psi(p)=m\psi(p)\end{align}
for $p\in \M$ and
$\|\psi\|=\sqrt{\sk{\psi,\psi}}<\infty$,
with the scalar product
\begin{equation}\label{eq:sk_M}
\sk{\phi,\psi}=
\int_\M\overline{\phi(p)}\psi(p)\,\frac{i_p(d^4 p)}{m}.
\end{equation}
Let $\CM\subseteq \HM$ denote the subspace of all functions $\psi\in\HM$
that have a holomorphic continuation
$\Psi:\M_{\C}\to\C^4$ fulfilling the bound
\begin{equation}
\label{bound 2}
\exists\alpha>0\;\forall n\in\N:\;\|\psi\|_{\M,\alpha, n}<\infty
\end{equation}
with
\begin{equation}
\label{bound 2a}
\|\psi\|_{\M,\alpha, n}:=
\sup_{p\in\M_{\C}} |p|^{n-1}e^{-\alpha |\im \vec p|}|\Psi(p)|
.
\end{equation}
\end{definition}
It is shown in Corollary \ref{cor:CM_dense} below that $\CM$ is dense in $\HM$.
It turns out that the norms (\ref{bound 2a}) involving \db 
only the \ed 3-vector part $\im\vec p$, and not the 4-vector $\im p$, in the
exponent are more convenient to use\db, and regarding Lorentz invariance \ed
the particular choice makes no difference as we shall see in
Section~\ref{sec:poincare} below.

Note that the holomorphic continuation $\Psi$ is uniquely determined by
$\psi\in \CM$. Even more, it is already determined by its restriction to
any non-empty relatively to $\M$ open 
subset $U\subseteq \M$. In particular, the restriction of $\psi$
to $\M_+$ already determines its values on $\M_-$ and vice versa.
Furthermore, the condition $\psi(p)\in {\cal D}_p$ for all $p\in\M$ 
extends analytically to
\begin{align}
\slashed{p}\Psi(p)=m\Psi(p)\qquad \text{for all }p\in\M_{\C}.
\end{align}
Note further that the inner product (\ref{eq:sk_M}) 
is positive definite. This can be seen as follows.
For $p\in \M$, the facts $\phi(p)\in{\cal D}_p$ and $\gamma^0\slashed{p}=
\slashed{p}^*\gamma^0$ imply 
$\overline{\phi(p)}\slashed{p}=m\overline{\phi(p)}$ and thus
\begin{equation}
m\overline{\phi(p)}\gamma^\mu\psi(p)=
\overline{\phi(p)}\gamma^\mu\slashed{p}\psi(p)=
\overline{\phi(p)}(2p^\mu-\slashed{p}\gamma^\mu)\psi(p)
=\overline{\phi(p)}(2p^\mu-m\gamma^\mu)\psi(p).
\end{equation}
We conclude
\begin{equation}
\label{gamma to p}
m\overline{\phi(p)}\gamma^\mu\psi(p)=
p^\mu\overline{\phi(p)}\psi(p).
\end{equation}
Using (\ref{eq:ipd4p}), we get
\begin{equation}
\label{skalarprodukt 2}
\sk{\phi,\psi}=
\int_\M\overline{\phi(p)}\frac{p^0}{p^0}\psi(p)\,\frac{i_p(d^4 p)}{m}
=
\int_\M\overline{\phi(p)}\gamma^0\psi(p)\,\frac{i_p(d^4 p)}{p^0}
=m^2\int_\M\phi(p)^*\psi(p)\,\frac{d^3p}{(p^0)^2},
\end{equation}
which is positive for $\phi=\psi$ unless $\psi=0$ almost everywhere.\\

As it turns out the causal structure of solutions to the Dirac equation can be seen to emerge
from the following geometric property of $\M_\C$, proven in Lemma~\ref{lemma:
mass shell geometry} in the appendix:
\begin{align}
    |\im p^0|\le |\im \vec p| \qquad \text{for all }(p^0,\vec p)\in\M_\C.
\end{align}

Finally, it will be convenient to introduce a Hilbert space $\HThree$ for the
$3$-momentum representation as it allows to fall back on many classical results
about the standard Fourier transform.  In view of the Paley-Wiener theorem we
define also a subspace $\CThree\subset \HThree$ consisting of certain
real-analytic functions. 

\begin{definition}
\db We endow \ed $\HThree:=L^2(\R^3,\C^4)$ 
with the standard scalar product
\begin{equation}
\sk{\phi,\psi}=
\int_{\R^3}\phi^*(\vec p)\psi(\vec p)\,d^3\vec p
\end{equation}	
and the corresponding norm $\|\phi\|=\sqrt{\sk{\phi,\phi}}$.
Let $\CThree$ be the subspace of $\HThree$ consisting of all functions $\phi:\R^3\to\C^4$
that have a holomorphic continuation $\Phi:\C^3\to\C^4$
that fulfills the bound
\begin{equation}
\label{bound 1}
\exists\alpha>0\;\forall n\in\N_0:\;
\|\phi\|_{3,\alpha,n}<\infty,
\end{equation}
where we set (using the notation $a\vee b = \max\{a,b\}$)
\begin{equation}
\label{bound 1a}
\|\phi\|_{3,\alpha,n}:=\sup_{\vec p\in\C^3} (m\vee |\vec p|)^ne^{-\alpha |\im \vec p|}
|\Phi(\vec p)|.
\end{equation}
This is well-defined \db as \ed $\Phi$ is uniquely determined by $\phi$.
\end{definition}

\begin{lemma}\label{lem:CThree_dense}
$\CThree$ is dense in $\HThree$.  
\end{lemma}
\begin{proof}
By the classical Paley-Wiener Theorem \cite[Theorem IX.11]{reed_methods_1981} 
the Fourier transform $L^{2}(\R^{3},\C^{4})\to \HThree $ maps 
$\cC_c^{\infty}(\R^{3},\C^{4})$ bijectively onto $\CThree$. 
Because the Fourier transform is unitary and $\cC_c^{\infty}(\R^{3},\C^{4})$ 
is dense in $L^{2}(\R^{3},\C^{4})$ the claim follows.
\end{proof}

\subsection{Action of the Poincar\'e Group}\label{sec:poincare}

In our later analysis, Poincar\'e transformations
will prove to be very helpful in computations. First, we introduce the Lorentz
transformations on 4-spinors $\psi\in \C^4$ and on 
space-time points $x\in\R^4$. They
are specified by a pair $(S,\Lambda)$
with a spinor matrix $S\in\C^{4\times 4}$
and a matrix
$\Lambda=({\Lambda^\mu}_\nu)_{\mu,\nu=0,1,2,3}\in\R^{4\times 4}$
that fulfill
\begin{align}
\label{lorentz1}
{\Lambda^\mu}_\sigma g_{\mu\nu}{\Lambda^\nu}_\tau=g_{\sigma\tau},
\qquad
S^*\gamma^0 S=\gamma^0
\end{align}
and are related by
\begin{equation}
\label{relation Lambda S}
{\Lambda^\mu}_\nu\gamma^\nu=S^{-1}\gamma^\mu S.
\end{equation}
Space-time points $x\in\R^4$ and Dirac spinor fields $\psi:\R^{4}\to\C^4$ 
are transformed by 
\begin{align}
{x'}^\mu&={\Lambda^\mu}_\nu x^\nu,\\
\psi'(x')&=S\psi(x).
\end{align}
Let 
\begin{equation}
T=\frac{1}{\sqrt{2}}\left(\begin{array}{cc}1&1\\1&-1\end{array}\right)=T^{-1}
\in\C^{4\times 4}
\end{equation}
be the unitary transformation matrix from the standard representation
of Dirac spinors to the Weyl spinor representation; see \cite[Appendix 1.A]{thaller:92}.
Transformations associated to the proper, orthochronous Lorentz group
$\operatorname{SO}^\uparrow(1,3)$ are parametrized by 
\begin{equation}\label{eq:lorentz-trafo}
\operatorname{SL}(2,\C)\ni M\mapsto (S(M),\Lambda(M))\in 
\operatorname{GL}(4,\C)\times\operatorname{SO}^\uparrow(1,3),
\end{equation}
where
\begin{equation}
S(M)=
T^{-1}\left(\begin{array}{cc}M&0\\0&(M^*)^{-1}\end{array}\right)T
\in \C^{4\times 4},
\end{equation}
and
$\Lambda(M)\in\R^{4\times 4}$ is the unique matrix 
such that equation~(\ref{relation Lambda S})
holds for $S=S(M)$ and $\Lambda=\Lambda(M)$.
Note that the map (\ref{eq:lorentz-trafo}) is a group homomorphism.

Let $m_\Lambda:\R^4\to\R^4$, $x\mapsto x'=\Lambda x$, denote
multiplication with $\Lambda$.
Using equation~(\ref{relation Lambda S}), 
we observe the following Lorentz covariance 
relations for pull-back w.r.t.\ this map: 
\begin{align}
\label{eq:invarianz-i_gamma}
S^{-1}[m_\Lambda^*i_{\gamma}(d^4x')]S
&=
S^{-1}\gamma^\mu S\,m_\Lambda^*i_{e_\mu}(d^4x')
={\Lambda^\mu}_\nu\gamma^\nu 
i_{\Lambda^{-1}e_\mu}(d^4x)
=\gamma^\nu\,
i_{e_\nu}(d^4x)=i_{\gamma}(d^4x),
\\ 
\label{eq: invarianz-ipd4p}
m_\Lambda^*i_{p'}(d^4p')
&=i_{p}(d^4p).
\end{align}
Poincar\'e transformations act in a natural way on the spaces
defined in this section:
\begin{definition}[Translations and Lorentz Transformations]
\label{def: Poincar'e transformations}
  Let $\Sigma$ be a Cauchy surface and $y\in\R^4$. We define the translation maps:
\begin{align}
\label{eq: transl Sigma}
&T_\Sigma^{-y}:\CSigma\to\cC_{\Sigma-y},
&T_\Sigma^{-y}\psi(x)=\psi(x+y)
&\qquad\text{ for } x\in\Sigma-y;
\\
\label{eq: transl s}
&T_{A}^{-y}:\CA\to\cC_{A(\cdot+y)},
&T_{A}^{-y}\psi(x)=\psi(x+y)
&\qquad\text{ for } x\in\R^4;
\\
\label{eq: transl M}
&T_{\mathcal M}^{-y}:\CM\to\CM,
&T_{\mathcal M}^{-y}\phi(p)=e^{-ipy}\phi(p)
&\qquad\text{ for } p\in\mathcal{M}.
\end{align}
Furthermore, for any $M\in \operatorname{SL}(2,\C)$ associated with a proper,
orthochronous Lorentz transformation $\Lambda=\Lambda(M)$ and a spinor
transformation $S=S(M)$ as in (\ref{eq:lorentz-trafo}) we define the maps
\begin{align}
\label{eq: Lorentz Sigma}
&L_\Sigma^{(S,\Lambda)}:\CSigma\to\cC_{\Lambda\Sigma},
&L_\Sigma^{(S,\Lambda)}\psi(x)=S\psi(\Lambda^{-1} x)
&\qquad\text{ for } x\in\Lambda\Sigma;
\\
\label{eq: Lorentz s}
&L_{A}^{(S,\Lambda)}:\CA\to\cC_{\Lambda A(\Lambda^{-1}\cdot)},
&L_{A}^{(S,\Lambda)}\psi(x)=S\psi(\Lambda^{-1}x)
&\qquad\text{ for } x\in\R^4;
\\
\label{eq: Lorentz M}
&L_{\mathcal M}^{(S,\Lambda)}:\CM\to\CM,\quad
&L_{\mathcal M}^{(S,\Lambda)}\phi(p)=S\phi(\Lambda^{-1}p)
&\qquad\text{ for } p\in\mathcal{M}.
\end{align}
\end{definition} 

\begin{lemma}\label{lemma: Poincar'e transformations}
The six maps specified in Definition \ref{def: Poincar'e transformations}
are well-defined. More precisely, they take their values
in the spaces 
specified in formulas (\ref{eq: transl Sigma})--(\ref{eq: Lorentz M}).
They extend to unitary maps, also denoted by 
$T_\Sigma^{-y}:\HSigma\to\mathcal{H}_{\Sigma-y}$, 
$T_{A}^{-y}:\HA\to\mathcal{H}_{A(\cdot+y)}$,
$T_{\mathcal M}^{-y}:\HM\to\HM$,
$L_\Sigma^{(S,\Lambda)}:\HSigma\to\mathcal{H}_{\Lambda\Sigma}$,
$L_{A}^{(S,\Lambda)}:\HA\to\mathcal{H}_{\Lambda A(\Lambda^{-1}\cdot)}$, and
$L_{\mathcal M}^{(S,\Lambda)}:\HM\to\HM$, respectively.
\end{lemma}

The proof is given in Appendix~\ref{sec:appendix}. 

\subsection{Change of Gauge}

Another physically relevant transformation is the change of gauge in the
electrodynamic potential $A$. The transformation of the potential is defined as
\begin{align}
    A'_\mu(x) \mapsto A_\mu(x) + \partial_\mu \lambda(x)
\end{align}
for any scalar field $\lambda\in\cal C^\infty_c(\R^4,\R)$.

\begin{definition}[Gauge Transformation]
\label{def:gauge}
  Let $\lambda\in\cal C^\infty_c(\R^4,\R)$. We define 
\begin{align}
\label{eq:gauge}
&\Gamma_{\lambda}:\CA\to\cC_{A+\partial\lambda},
&\Gamma_{\lambda}\psi(x)=e^{-i\lambda(x)}\psi(x)
&\qquad\text{ for } x\in\R^4,
\end{align}
\end{definition} 

\begin{lemma}\label{lemma:gauge}
The maps specified in Definition~\ref{def:gauge}
are well-defined. More precisely, they take their values
in the spaces 
specified in formula (\ref{eq:gauge}).
They extend to unitary maps, also denoted by 
$\Gamma_\lambda:\HA\to\mathcal{H}_{A+\partial\lambda}$.
\end{lemma}

The proof is given in Appendix~\ref{sec:appendix}. 

\subsection{Generalized Fourier Transforms}\label{sec:free-sol}

Next we introduce the mentioned generalized Fourier transforms.
Their properties are collected in the
following main theorem. An immediate byproduct is an evolution
operator for the free Dirac equation, i.e., equation (\ref{eq:dirac-equation})
for $A=0$. In the following we write
\begin{equation}
px=p_\mu x^\mu= p^0x^0-p^1 x^1-p^2x^2-p^3x^3=p^0x^0-\vec p\cdot \vec x,
\quad x,p\in\C^4.
\end{equation}
The two different meanings of $p^{2}$ as second component of $p\in\C^4$ and 
$p^2=p_\mu p^\mu$ will be unambiguous given the context. 

\begin{theorem}[Generalized Fourier Transforms and Free Dirac Evolution]
\label{thm:F_maps}
\quad
\begin{enumerate}
\item
  For all $I,J,K$ being placeholders for the symbols $3,\M,\sol$ or any Cauchy surface $\Sigma$ there are unique unitary maps
$\F_{IJ}:\cal H_{J}\to\cal H_{I}$
  with the following properties:
  
  \begin{enumerate}[(i)]
  \item $\F_{II} = \mathrm{id}_{\cal H_{I}}$.
  \item $\F_{IJ} \F_{JK} = \F_{IK}$.
  \item $\F_{IJ}$ maps $\cC_{J}$ bijectively onto $\cC_{I}$.
  \item The maps $\F_{3\M}$, $\F_{\M 3}$, $\F_{\M\Sigma}$,
$\F_{\sol\M}$, and $\F_{\Sigma\sol}$ are characterized as follows:
\begin{align}
\label{def F3M}
(\F_{3\M}\psi)(\vec p)&=m \frac{\psi(p_{+}(\vec p))-\psi(p_{-}(\vec p))}{E(\vec p)}
&&\text{for }\psi\in\HM, \vec p\in\R^{3};
\\
\label{def FM3}
\db (\F_{\M 3}\psi)(p)&=\frac{\slashed p+m}{2m}\gamma^{0}\psi(\vec p)
&&\text{for }\psi\in\HThree,  p=(p^0,\vec p)\in\M;\ed
\\
\label{def FMSigma}
\db(\F_{\M\Sigma}\psi)(p)&=\frac{\slashed p+m}{2m}(2\pi)^{-3/2}
\int_\Sigma e^{ipx}\,i_\gamma(d^4x)\,\psi(x)
&&\text{for } \psi\in\CSigma, p\in\M;
\\
\label{def FsM}
(\F_{\sol\M}\psi)(x)&=\frac{(2\pi)^{-3/2}}{m}\int_\M e^{-ipx}\, i_p(d^4p)\,\psi(p)
&&\text{for } \psi\in\CM,x\in\R^{4};
\\
\label{def FSigmas}
(\F_{\Sigma\sol}\psi)(x)&=\psi(x)
&&\text{for } \psi\in\Cs, x\in\Sigma.
\end{align}

\end{enumerate}

\item For $\psi\in\CSigma$,
the function $\F_{\sol\Sigma}\psi$
is supported in $\operatorname{supp}\psi+\causal$.

\item For any symbol $I$ among $3,\M,\Sigma,\sol$ the space $\cC_{I}$ is dense in $\cal H_{I}$.

\end{enumerate}

\end{theorem}

The free Dirac evolution between Cauchy surfaces $\Sigma$ and $\Sigma'$ is
given by the unitary map $\F_{\Sigma'\Sigma}$. 
In physicists' notation the formal integral kernel of $\F_{\sol\Sigma}$ is usually
called the propagator of the free Dirac equation.
The maps $\F_{\M\Sigma}$ and $\F_{\sol\M}$ commute
with Poincar\'e transformations in the following sense:

\begin{theorem}[Compatibility with
    Poincar\'e Transformations]
\label{thm: compatibility Poincar'e}
For any translation vector $y\in\R^4$ and any Lorentz transformation
associated with $\Lambda=\Lambda(M)$, $S=S(M)$ with 
\db $M\in\operatorname{SL}(2,\C)$, cf. Definition~\ref{def: Poincar'e transformations}\ed, the following compatibility relations
hold true.
\begin{align}
\label{eq:translation-M-Sigma}
T_{\mathcal M}^{-y} \F_{\M\Sigma}
&=\F_{\M, \Sigma-y}  T_\Sigma^{-y},\\
\label{eq:translation-s-M}
T_{\sol}^{-y} \F_{\sol\M}
&=\F_{\sol\M}  T_\M^{-y},\\
\label{eq:Lorentz-M-Sigma}
L_{\mathcal M}^{(S,\Lambda)} \F_{\M\Sigma}
&=
\F_{\M, \Lambda\Sigma}  L_\Sigma^{(S,\Lambda)},
\\
\label{eq:Lorentz-s-M}
L_{\sol}^{(S,\Lambda)} \F_{\sol\M}
&=
\F_{\sol\M}  L_{\mathcal M}^{(S,\Lambda)}
\end{align}
\end{theorem}

The proof of Theorem~\ref{thm:F_maps} is given in Section
\ref{sec:generalized-fourier}, and the proof of Theorem 
\ref{thm: compatibility Poincar'e} is given in Appendix~\ref{sec:appendix}. 
While the latter is straight-forward the former needs several
technical lemmas, some of which are phrased in the following subsection; the
remaining technical lemmas and proofs are given in Section~\ref{sec:gen-fou-trafo}.
Note that beside the regularity information contained in the spaces $\CSigma$ and $\Cs$ 
Theorem~\ref{thm:F_maps} (b) makes precise the causal structure of the support properties
of solutions of the free Dirac equation.
These support properties are
controlled by Paley-Wiener techniques.  The standard Paley-Wiener theorem treats
the case in which both position and momentum spaces are flat. To apply this
theorem to curved Cauchy surfaces we employ a family of projections from a
particular Cauchy surface to a flat position space; see proof of
Lemma~\ref{lemma: FMSigma}.  On the momentum space side, the mass-shell $\M$ is
analytically continued to a complex $3$-dimensional manifold $\M_\C$ which is
also projected onto a complexified 3-momentum space $\C^3$.  The Paley-Wiener
theorem is then \db applied \ed not \old{used for}\db to the \ed holomorphic
functions on $\M_\C$ directly but to their appropriate projections.  Tools from
complex analysis help to control the projected functions quantitatively, in
particular close to the ramification set of the projection; see proof of
Lemma~\ref{lem:FM3_F3M}. The following section provides an overview of the most
important bounds, and in Section~\ref{sec:sobolev} we introduce appropriate
Sobolev norms that are also helpful to control regularity of solutions of the
Dirac equation subject to an external potential $A$; see
Section~\ref{sec:dirac-et-overview}.

\subsubsection{Paley-Wiener Bounds}\label{sec:PWbounds}

Let $K\subset\R^4$ be a compact set and $0\leq V<1$.
We define $\cal S(K,V)$ to be the set of all Cauchy surfaces $\Sigma$
with $\Sigma\cap K\neq\emptyset$ and $\sup_{\vec x\in \R^3}|\nabla
t_\Sigma(\vec x)|\leq V$; cf. (\ref{eq:parametrize Sigma}).
For $\Sigma\in\cal S(K,V)$ let $\cal C_\Sigma(K)$ denote the set of all wave
functions $\psi\in\CSigma$ supported in $K\cap\Sigma$. For such $\psi$ and
$n\in\N_0$ we define
\begin{align}
    \Vert \psi \Vert_{\Sigma,K,n}
    =
    \sup_{\vec x\in\R^3} \sum_{|\beta|\leq n}\left| 
        D^\beta \psi(t_\Sigma(\vec x),\vec x)
    \right|,
\end{align}
where the differential operator $D^\beta$ for a multi-index $\beta\in\N_0^3$
acts on $\vec x\in\R^3$.

\begin{theorem}[Paley-Wiener Bounds for Cauchy Surfaces]
    \label{lemma: FMSigma}
For any $K,V$ as above, any $\Sigma\in\mathcal S(K,V)$, any 
$\psi\in\cal C_\Sigma(K)$, any positive number $\alpha$ such that $\alpha>\sqrt
2\sup_{x\in K}|x|$,
and any $n\in\N$ one has
\begin{align}
    \|\F_{\M\Sigma}\psi\|_{\M,\alpha, n}
    \leq
    \constl{c:pwbound} \Vert \psi \Vert_{\Sigma,K,n}
\end{align}
with some positive constant $\constr{c:pwbound}=\constr{c:pwbound}(K,V,n,\alpha,m)$.
\end{theorem}

This theorem is proven in the proof of Lemma~\ref{lemma: FMSigma}. Furthermore,
we give the following bounds which are useful in switching between the spaces $\CM$ and $\CThree$.

\begin{theorem}[Bounds on $\CM$ and $\CThree$]
    \label{eq:M3bounds}
    For any $\alpha>0$, $\epsilon>0$, and $n\in\N$ the following bounds hold\ed
\begin{align}
\label{bound F3M-v1}
\|\F_{3\M}\psi\|_{3,\alpha, n-1}
&\le 
\constr{4-v1}\|\psi\|_{\M,\alpha, n}
&&\db \text{for }\psi\in\CM,\ed
\\
\label{bound F3M}
\|\F_{3\M}\psi\|_{3,\alpha, n}
&\le 
\constr{4}\|\psi\|_{\M,\alpha+\epsilon, n}
&&\text{for }\psi\in\CM,
\\
\label{bound FM3}
\|\F_{\M 3}\phi\|_{\M,\alpha, n}
&\le 
\constr{cfm3}\|\phi\|_{3,\alpha, n}
&& \text{for }\phi\in\CThree,
\end{align}
with positive constants $\constr{4-v1}=\constr{4-v1}(n,\alpha,m)$,
$\constr{4}=\constr{4}(n,\alpha,\epsilon,m)$, and
$\constr{cfm3}=\constr{cfm3}(n,m)$.
\end{theorem}

This theorem is proven in the proof of Lemma~\ref{lem:FM3_F3M}. \\

\subsubsection{Sobolev Norms}\label{sec:sobolev}

On the one hand, the norms $\|{\cdot}\|_{\M,\alpha,n}$, introduced in
(\ref{bound 2a}) are well adapted to Paley-Wiener arguments and are therefore
useful for the analysis of support properties. 
On the other hand, Sobolev norms turn out to be
more convenient for the analysis of regularity. Now we introduce a
version of Sobolev norms well suited for the analysis of the Dirac equation.

\begin{definition}\label{def:sobolev}
For $n\in\N_0$, let $\cH_{\M,n}$ denote the vector space of all $\psi\in \HM$
such that $p^\beta \psi\in  \HM$ for any multi-index $\beta\in\N_0^4$ with
$|\beta|\le n$.  Here
$p^\beta:=p_0^{\beta_0}p_1^{\beta_1}p_2^{\beta_2}p_3^{\beta_3}$, where $p_j$
stands for the multiplication operator with $p_j$, $p\in\M$.  We endow
$\cH_{\M,n}$ with the norm
\begin{equation}
\label{def-normMn}
\|\psi\|_{\M,n}^2:=
\sum_{\substack{\beta\in\N_0^4\\|\beta|\le n}}\|p^\beta\psi\|^2
=
\sum_{\substack{\beta\in\N_0^4\\|\beta|\le n}}\int_{\M}|p^\beta|^2
\overline{\psi(p)}\psi(p)\frac{i_p(d^4p)}{m}.
\end{equation}
Given a Cauchy surface $\Sigma$, for a placeholder $I$ 
standing for $\sol$ or $\Sigma$, we define the normed
space
\begin{align}
    \cH_{I,n} := \F_{I\M}\left[\cH_{\M,n}\right],
  \qquad
  \|\psi\|_{I,n}:=\|\F_{\M I}\psi\|_{\M,n}
\end{align}
and for any $j=0,1,2,3$ the bounded operator
\begin{align}
  \partial_j: \cH_{I,n+1}\to \cH_{I,n},
  &
  \quad \psi \mapsto -i\F_{I\M}
  \, p_j \, \F_{\M I} \psi.
\end{align}
\end{definition}

We remark that for any placeholder $I$, standing for
$\M,\sol$, or $\Sigma$, the space
$(\cH_{I,n},\|{\cdot}\|_{I,n})$ is a Hilbert space containing
$\cC_I$ as a dense subspace. Furthermore, the multiplication operator
$p^\beta:
(\cH_{\M,n},\|{\cdot}\|_{\M,n})\to
(\cH_{\M,n-|\beta|},\|{\cdot}\|_{\M,n-|\beta|})$
is bounded. Note that the restriction of $\partial_j:\cH_{\sol,n+1}\to \cH_{\sol,n}$
to $\Cs$ is the differential operator
$\partial_j\psi(x)=\frac{\partial}{\partial x^j}\psi(x)$.
For $\psi\in\cH_{\M,n+1}$, $n\geq 0$, and
$\widehat\psi=\F_{\M\sol}\psi$ one \old{observes}\db has\ed
\begin{align}
  \Vert \partial_j \psi \Vert_{\sol,n}
  =
  \left\Vert p_j \widehat\psi \right\Vert_{\M,n}.
\end{align}

Thanks to the free Dirac equation, the restriction of $\partial_j:\cH_{\Sigma,n+1}\to\cH_{\Sigma,n}$ to
$\cC_\Sigma$ is also a differential operator. For $\vec x=(x^1,x^2,x^3)\in\R^3$ it takes the form
\begin{align}\label{eq:partial-k}
  \partial_j \psi(x)
  =
  -i\left(
    \sum_{k=1}^3 \alpha^\Sigma_{jk}(x)D_k
    + \beta^\Sigma_{j}(x)
  \right) \psi(x),
  \qquad
  x=(t_\Sigma(\vec x),\vec x),
\end{align}
with some smooth functions $\alpha^\Sigma_{jk},
\beta^\Sigma_{j}:\Sigma\to\C^{4\times4}$ depending only on the geometry of
$\Sigma$ and
\begin{align}\label{eq:D-k}
  D_k \psi(t_\Sigma(\vec x),\vec x) 
  =
  \frac{\partial}{\partial x^k}
  \left(
    \psi(t_\Sigma(\vec x),\vec x)
  \right) = \partial_k \psi(t_\Sigma(\vec x),\vec x) + \frac{\partial
  t_\Sigma(\vec x)}{\partial x^k} \partial_0
  \psi(t_\Sigma(\vec x),\vec x).
\end{align}

The following lemma shows that pointwise evaluation for elements of
$\cH_{\sol,n}$ makes sense whenever $n\ge 2$.

\begin{lemma}[Pointwise Evaluation]
\label{lemma: Pointwise evaluation}
For $n\in\N$ with $n\ge 2$ and $x\in\R^4$, the evaluation map 
$\delta_x:\Cs\to \C^4$, $\psi\mapsto \psi(x)$,
extends to a bounded linear map $\delta_x:
(\cH_{\sol,n},\|{\cdot}\|_{\sol,n})\to\C^4$, also denoted
by $\delta_x:\psi\mapsto \psi(x)$.
\end{lemma}
\begin{proof}
Given
$\phi\in\CM$ and $n\ge 2$, 
using the Definition~(\ref{def FsM}) of $\F_{\sol\M}$,
the Cauchy-Schwarz-inequality, and (\ref{skalarprodukt 2}),
one has 
\begin{align}
|\F_{\sol\M}\phi(x)|^2&=
\left|\frac{(2\pi)^{-3/2}}{m}\int_\M e^{-ipx}\phi(p)\, i_p(d^4p)\right|^2
\cr&
\le\frac{(2\pi)^{-3}}{m^2}
\int_\M 
|q^0|^{-2} \,\frac{i_q(d^4q)}{q^0}\
\int_\M 
|(p^0)^2\phi(p)|^2\, \frac{i_p(d^4p)}{p^0}
\cr&\le
\constl{const:delta}^2\|(p_0)^2\phi\|^2
\le\constr{const:delta}^2\|\phi\|_{\M,n}^2
\end{align}
with some positive constant $\constr{const:delta}=\constr{const:delta}(m)$.
Setting $\phi=\F_{\M\sol}\psi$ for any given $\psi\in\Cs$,
it follows
\begin{equation}
|\psi(x)|\le \constr{const:delta}\|\phi\|_{\M,n}=
\constr{const:delta}\|\psi\|_{\sol,n}.
\end{equation}
\old{When passing to the completion $\cH_{\sol,n}$, the claim follows.}
The claim then follows by passing to the completion in $\cH_{\sol,n}$.
\end{proof}

\subsection{Existence, Uniqueness, and Causal Structure}
\label{sec:main-results}

The next theorem is about 
the well-posedness of the initial value problem corresponding
to (\ref{eq:dirac-equation}).
For a given  $\psi\in\CA$ and a
Cauchy surface $\Sigma$, we denote the
restriction of $\psi$ to $\Sigma$ by
$\psi|_{\Sigma}\in\cC^\infty_c(\Sigma,\C^4)$.\ed
\begin{theorem}[Initial Value Problem and Support]
    \label{dirac-existence-uniqueness}
  Let $\Sigma$ be a Cauchy surface and
  $\chi_\Sigma\in\cC_{c}^{\infty}(\Sigma,\C^{4})$ be given initial data.
  Then the following is true: 
  \begin{enumerate}[(i)]
    \item There is a $\psi\in\CA$ such that $\psi|_{\Sigma}=\chi_\Sigma$
       and $\supp \psi \subseteq \supp \chi_\Sigma + \operatorname{Causal}$.
    \item Suppose $\widetilde\psi\in\cC^\infty(\R^4,\C^4)$ solves the Dirac
        equation
        (\ref{eq:dirac-equation}) for initial data $\widetilde\psi|_\Sigma=\chi_\Sigma$. 
        Then $\widetilde\psi=\psi$.
  \end{enumerate}
\end{theorem}

This theorem gives rise to the following definition.

\begin{definition}[Evolution Operator]
  Let $\Sigma'$ be another Cauchy surface. Given $\chi_{\Sigma}$ with the
  corresponding $\psi\in\CA$ as above,
  we define the Dirac evolution from $\Sigma$ to $\Sigma'$ by
  \begin{align}\F^A_{\Sigma'\Sigma}\chi_{\Sigma}:=\psi|_{\Sigma'}\end{align}
  which yields a map $\F^A_{\Sigma'\Sigma}:\cC^{\infty}_{c}(\Sigma,\C^{4})\to\cC^{\infty}_{c}(\Sigma',\C^{4})$.
\end{definition}

\db As a direct consequence of Theorem~\ref{dirac-existence-uniqueness} we
infer:\ed

\begin{theorem}[Unitary Evolution]
\label{thm: Hilbert space time evolution}
The map
$\F^A_{\Sigma'\Sigma}:\cC^{\infty}_{c}(\Sigma,\C^{4})\to\cC^{\infty}_{c}(\Sigma',\C^{4})$
extends uniquely to a unitary map
$\F^A_{\Sigma'\Sigma}:\HSigma\to\cH_{\Sigma'}$.
\end{theorem}

%The proof \db of \ed Theorem~\ref{dirac-existence-uniqueness} consists of two
%steps: The first step deals with the free Dirac equation, i.e.,
%(\ref{eq:dirac-equation}) for $A=0$. 
As discussed in the introduction, there are several different strategies of proof for
Theorem~\ref{dirac-existence-uniqueness} and \ref{thm: Hilbert space
time evolution} in the literature. In this work we will give a proof
with the help of the just introduced generalized
Fourier transforms. We recall that the
collected results about these Fourier transforms in Section~\ref{sec:free-sol}
already include a proof of Theorem~\ref{dirac-existence-uniqueness}  and 
\ref{thm: Hilbert space time evolution} in the case of
$A=0$. With a Picard-Lindel\"of argument this result can readily be extended to include
an external vector potential in the Dirac evolution.
We shall use this opportunity to introduce a convenient interaction picture
adapted to Cauchy surfaces; see Section~\ref{sec:interaction}.
The main ingredient in the switching from the Schr\"odinger picture to this
interaction picture are again the generalized Fourier
transforms.  In the interaction picture, the
Dirac equation is rephrased in terms of an ordinary differential equation for
functions taking values in Sobolev spaces, introduced in
Definition~\ref{def:sobolev}, composed of solutions of the free Dirac equation.
The Picard-Lindel\"of theorem then yields existence and uniqueness of solutions,
see Lemma~\ref{lem:interaction-existence-uniqueness}, while regularity of
solutions is analyzed with the help of a version of Sobolev's lemma adapted to
Cauchy surfaces; see Lemma~\ref{lemma: regularity}.  The support properties of
the free Dirac evolution and the
Picard-Lindel\"of iteration imply the support properties of the solutions of the
Dirac equation with external potential.

%\section{The Free Dirac Evolution}\label{sec:free-dirac}
%
%\db In this section we will develop \ed a solution theory to the free Dirac equation 
%$i\slashed\partial\psi=m\psi$. 
%The goal is to construct an evolution operator that 
%transports initial data given on one Cauchy surface to another. 
%This \db construction \ed will depend on our definition of Cauchy surfaces, 
%see Definition \ref{def:cauchy-surface}, in a crucial way.
%
%\db Before we get to generalized Fourier transforms in
%Section~\ref{sec:free-sol}
%and the construction of the Dirac evolution,
%in Sections~\ref{sec:spaces-cauchy-sol} and
%\ref{sec:spaces-momentum} we introduce further vector spaces
%$\CSigma\subset \HSigma$, $\CThree\subset\HThree$, $\CM\subset\HM$
%and $\HA\supset\CA$  in addition to the spaces $\CA$ and $\HSigma$
%defined in Definitions
%\ref{def:CA} and \ref{def:HSigma}. Furthermore, we discuss the action of the
%Poincar\'e group in Section~\ref{sec:poincare}
%which will turn out to be an essential tool in our proofs.
%\ed

\subsection{An Interaction Picture on Cauchy Surfaces}\label{sec:interaction}

\db As discussed in the previous section it can be useful to switch to an
interaction picture in order to treat the interaction with the external
potential.  For this we introduce a family of Cauchy surfaces
$(\Sigma_t)_{t\in\R}$ driven by a family of normal vector fields 
$(v_t n_t|_{\Sigma_t})_{t\in\R}$, where $n:\R^4\times\R\to \R^4, x\mapsto
n_t^\mu(x)$ and $v:\R^4\times\R\to\R, v:(x,t)\mapsto v_t(x)$ are smooth
functions. 
For $x\in\Sigma_t$ the vector $n_t(x)$ denotes the future-directed unit-normal
vector to $\Sigma_t$ and $v_t$ the corresponding normal velocity of the flow of
Cauchy surfaces.
In particular, given 
initial values $x_0\in\Sigma_0$, the solutions of the ODE
\begin{align}
    \dot x_t^\mu = v_t(x_t)n^\mu_t(x_t),
    \qquad
    \mu=0,1,2,3,
\end{align}
give rise to trajectories $t\mapsto x_t$ with $x_t\in\Sigma_t$ for all $t\in\R$.
Furthermore, we define the set 
$\mathbf{\Sigma}=\left\{(x,t) \in \R^4\times\R \, | \, x\in\Sigma_t \right\}$.
In case, the following conditions are satisfied:
\begin{itemize}
    \item $v_t(x) > 0$ for all $(x,t)\in\mathbf{\Sigma}$;
    \item the projection $F:\mathbf{\Sigma}\to \R^4, (x,t)\mapsto x$ is a
        diffeomorphism,
\end{itemize}
we call $\mathbf{\Sigma}$ a future-directed foliation of space-time and define
$(y,\tau(y)) := F^{-1}(y)$ for $y\in\R^4$ for which
\begin{equation}
\partial_\mu \tau(x) = (n_{\tau(x)})_\mu(x) \, v_{\tau(x)}(x)^{-1}
    \label{eq:partial-tau}
\end{equation}
holds. 
Though defining $n$ and $v$ on $\mathbf \Sigma$ would suffice, it is sometimes
convenient to have them on the whole space $\R^4\times\R$.
A simple example of a foliation of space-time is given by
$\Sigma_t=\Sigma + te_0$ for $t\in\R$.
The following lemma describes the transition from the Dirac equation in the
Sch\"odinger picture to an interaction picture associated to \db the given family
of hypersurfaces \ed and
vice versa.  It is proven in Section~\ref{sect: class sol}, below.
\begin{theorem}[Equivalence of the Schr\"odinger Picture and the Interaction Picture]
\label{lemma: equiv}
Consider a future-directed foliation $\mathbf{\Sigma}$, the Cauchy surface
$\Sigma=\Sigma_{t=0}$, and let $\chi_\Sigma\in\CSigma$.
\begin{enumerate}
\item
Assume that $\psi\in\CA$ fulfills the initial condition
$\psi|_\Sigma=\chi_\Sigma$.
Define \db $\phi_t=\F_{0\Sigma_t}\psi|_{\Sigma_t}\in \cC_0$ \ed for all $t\in\R$.
Then the function $\phi:\R^4\times\R\to\C^4$,
$\db (x,t)\mapsto\ed\phi(x,t)=\phi_t(x)$ is smooth. It fulfills
the initial condition
\begin{align}
&\phi_0=\F_{\sol\Sigma}\chi_\Sigma
\label{eq:initial cond interaction picture}
\end{align}
and the following evolution equation for all $t\in\R$ and $x\in\R^4$:
\begin{align}
&i\frac{\partial}{\partial t} \phi_t(x)
=
L_t \phi_t(x)
\qquad
\text{ with }
\qquad
\db L_t := \F_{\sol \Sigma_t} ( v_t \slashed n_t \slashed A ) \F_{\Sigma_t \sol} \ed
:\Cs\selfmaps.
\label{eq:interaction-evolution-Cs}
\end{align}
Here, $(v_t\slashed n_t\slashed A):
\cC_{\Sigma_t}\selfmaps$ is understood as a multiplication operator
\begin{align}
    (v_t\slashed n_t\slashed A)\xi (x) = v_t(x)\slashed n_t(x)\slashed A(x)\xi(x),
    \qquad\text{for }\xi\in\cC_{\Sigma_t}, x\in\Sigma_t.
\end{align}
Furthermore, there is a compact set $K\subset\R^4$ such that
for all $t\in\R$ the function $\phi_t\in\Cs$ is supported in $K+\causal$.
Finally, one has $\psi(x)=\phi(x,\tau(x))$ for all $x\in\R^4$.
\item
Conversely, let
$\phi:\R^4\times\R\to\C^4$,
$\phi(x,t)=\phi_t(x)$ be a smooth function, supported
in  $(K+\causal)\times\R$ for some compact set $K\subseteq\R^4$.
Assume that $\phi_t\in\Cs$ for all $t\in\R$, and that $\phi$ fulfills
the evolution equation
(\ref{eq:interaction-evolution-Cs})
and the initial condition (\ref{eq:initial cond interaction picture}).
Let
\begin{equation}
\psi:\R^4\to\C^4,\quad
\psi(x):=\phi(x,\tau(x)).
\label{eq:psi-from-phi}
\end{equation}
Then $\psi\in\CA$ and $\psi|_\Sigma=\chi_\Sigma$. Finally, one has
\db $\phi_t=\F_{0\Sigma_t}\psi|_{\Sigma_t}$ \ed for all $t\in\R$.
\end{enumerate}
\end{theorem}

%Conversely, note that any smooth solution $\phi$ of the 
%initial value problem 
%(\ref{eq:initial cond interaction picture}), 
%(\ref{eq:interaction-evolution-Cs})
%in the classical sense
%as described in Theorem~\ref{lemma: equiv} (a)
%is also a solution of the initial value problem
%(\ref{eq:ivp}) for any $n\in\N$ in the Hilbert space sense.

\section{Proofs}
\label{sec:proofs}

In this last section we provide the remaining technical proofs of the claims in
Section~\ref{sec:review}. It is split in two parts. The first part, given in
Section~\ref{sec:gen-fou-trafo}, concerns the generalized Fourier transforms. The
second part, given in Section~\ref{sec:dirac-et-overview}, concerns the
solution theory of the Dirac equation.

\subsection{Generalized Fourier Transforms}
\label{sec:gen-fou-trafo}

\subsubsection{Properties of the Maps $\F_{3\M}$ and $\F_{\M 3}$}

The following lemma extends Therorem~\ref{eq:M3bounds}.

\begin{lemma}\label{lem:FM3_F3M}
The maps $\F_{3\M}$ and $\F_{\M 3}$ are well-defined
unitary operators. They are inverse to each other.
\db Furthermore, one has $\F_{3\M}[\CM]=\CThree$, 
$\F_{\M 3}[\CThree]=\CM$, and  
for any $\alpha>0$, $\epsilon>0$, and $n\in\N$ the following bounds hold\ed
\begin{align}
\label{bound F3M-v1}
\|\F_{3\M}\psi\|_{3,\alpha, n-1}
&\le 
\constl{4-v1}\|\psi\|_{\M,\alpha, n}
&&\db \text{for }\psi\in\CM,\ed
\\
\label{bound F3M}
\|\F_{3\M}\psi\|_{3,\alpha, n}
&\le 
\constl{4}\|\psi\|_{\M,\alpha+\epsilon, n}
&&\text{for }\psi\in\CM,
\\
\label{bound FM3}
\|\F_{\M 3}\phi\|_{\M,\alpha, n}
&\le 
\constl{cfm3}\|\phi\|_{3,\alpha, n}
&& \text{for }\phi\in\CThree,
\end{align}
with positive constants $\constr{4-v1}=\constr{4-v1}(n,\alpha,m)$,
$\constr{4}=\constr{4}(n,\alpha,\epsilon,m)$, and
$\constr{cfm3}=\constr{cfm3}(n,m)$.
\end{lemma}

\begin{proof}[Proof of Lemma~\ref{lem:FM3_F3M}]
\textsc{We show first that $\F_{3\M}$ is an isometry.}
We calculate for $\psi\in\HM$,
using that for $\vec p\in\R^3$, the vectors
$\psi(p_+(\vec p))\in{\cal D}_{p_+(\vec p)}$
and $\psi(p_-(\vec p))\in{\cal D}_{p_-(\vec p)}$
are orthogonal:
\begin{align}
\nonumber
\|\F_{3\M}\psi\|^2
&=
m^2 \int_{\R^3}|\psi(p_+(\vec p))-\psi(p_-(\vec p))|^2\frac{d^3{\vec p}}{E(\vec p)^2}
=
m^2 \int_{\R^3}(|\psi(p_+(\vec p))|^2+|\psi(p_-(\vec p))|^2)\frac{d^3{\vec p}}{E(\vec p)^2}
\\
&=
m^2 \int_\M|\psi(p)|^2\frac{d^3p}{(p^0)^2}
=\|\psi\|^2,
 \end{align}
where we have used equation~(\ref{skalarprodukt 2}) in the last step.\\

\textsc{Next, we show that $\F_{3\M}$ and $\F_{\M 3}$ are inverse to each
other.}
Consider the reflection $r:\M\to\M$,
$r(p^0,\vec p)=(-p^0,\vec p)$.
For $\psi\in\HM$ and $p=(p^0,\vec p)\in\M$, we get the following,
using the definition (\ref{eq:Pp}) of $P(p)$.
\begin{equation}
\F_{\M 3}\F_{3\M}\psi(p)
=\frac{\slashed{p}+m}{2}\gamma^0 \frac{\psi(p_+(\vec p))-\psi(p_-(\vec p))}{E(\vec p)}
=P(p) (\psi(p)-\psi(r(p)))
=\psi(p),
\end{equation}
where we have used that $P(p)$ acts as identity on ${\cal D}_p$ and as zero on
${\cal D}_{r(p)}$. Conversely we get for $\phi\in\HThree$:
\begin{equation}
\F_{3\M}\F_{\M 3}\phi(\vec p)
=P_+(\vec p)\phi(\vec p)+
P_-(\vec p)\phi(\vec p)=\phi(\vec p).
\end{equation}
Because $\F_{3\M}$ is an isometry, 
we conclude that $\F_{3\M}$ and $\F_{\M 3}$
are unitary maps.\\

\textsc{Now, we show (\ref{bound F3M-v1}), (\ref{bound F3M}) and
$\F_{3\M}[\CM]\subseteq\CThree$. } Let $\psi \in\CM$.
By definition, $\psi$ has a holomorphic extension $\Psi:\M_\C\to \C^4$
that fulfills the bound (\ref{bound 2}).
We extend the reflection $r:\M\to\M$ 
to the biholomorphic map $r:\M_\C\to\M_\C$,
$r(p^0,\vec p)=(-p^0,\vec p)$ and consider the ramification
set $Z=\{p\in\M_\C:\;
p^0=0\}$ consisting of fixed points of $r$. $Z$ is a complex submanifold
of $\M_\C$ of codimension 1; in particular it has no singular points.
The holomorphic map
\begin{equation}
\label{def chi}
\chi:\M_\C\setminus Z\to \C^4, \quad 
\chi(p)=m\frac{\Psi(p)-\Psi(r(p))}{p^0}
\end{equation}
is locally bounded near any point in $Z$, 
because the numerator $\Psi-\Psi\circ r$
vanishes on $Z$, and the denominator 
$\M_\C\ni p\mapsto p^0$ vanishes of first order
on $Z$. By Riemann's extension theorem, the map $\chi$ extends to a holomorphic
map on the whole set $\M_\C$. We denote this extension
also by $\chi:\M_\C\to \C^4$.
Now consider the projection $\pi:\M_\C\to\C^3$, $\pi(p^0,\vec p)=\vec p$,
and its set of branching points $\pi[Z]=\{\vec p\in\C^3:\;\vec p^2+m^2=0\}$.
Note that $\pi^{-1}[\pi(p)]=\{p,r(p)\}$  holds for any $p\in\M_\C$,
and that $\pi[Z]$ is also a submanifold
of $\C^3$ of complex codimension 1; in particular it has also no singular points.
Since $\chi\circ r=\chi$, there is a map $\Phi:\C^3\to\C^4$ such
that $\Phi\circ\pi=\chi$. Obviously $\Phi$ is holomorphic
outside the branching points, i.e.
on  $\C^3\setminus\pi[Z]$, and it is locally bounded near any branching 
point  
$\vec p \in \pi[Z]$. Using Riemann's extension theorem again, we see that
$\Phi$ is holomorphic on its whole domain $\C^3$.
Comparing definitions (\ref{def F3M}) and (\ref{def chi}),
we see that $\Phi:\C^3\to\C^4$ is a holomorphic extension 
of $\phi:=\F_{3\M}\psi:\R^3\to\C^4$.
To finish the proof of $\F_{3\M}\psi \in \CThree$, it 
remains to show that $\Phi$ fulfills the bound
(\ref{bound F3M});
recall the definition (\ref{bound 1})/(\ref{bound 1a}).
Take $\alpha>0$ such (\ref{bound 2})/(\ref{bound 2a}) holds.
Let $n\in\N$.
Using 
the definition (\ref{def chi}) of $\chi$, we get
the following for all $p\in\M_\C\setminus Z$:
\begin{equation}
\label{p0chi}
|p^0\chi(p)|\le 2m|p|^{-(n-1)} 
e^{\alpha |\im\vec  p|}\|\psi\|_{\M,\alpha,n}
\le 2m(m\vee |\vec p|)^{-(n-1)} 
e^{\alpha |\im\vec  p|}\|\psi\|_{\M,\alpha,n}
\end{equation}
For the last step, we have used $m\le |p|$ from (\ref{p0 bound 5}) in
Lemma~\ref{lemma: mass shell geometry} 
and $|\vec p|\le |p|$. We distinguish two cases:\\
{\it Case 1, ``locations far from the ramification set'':} $|p^0|\ge m/12$.
On the one hand,
(\ref{p0chi}) implies in this case the following.
\begin{equation}
\label{ineq p0 big-v1}
|\chi(p)|\le 24(m\vee|\vec p|)^{-(n-1)}
e^{\alpha |\im \vec  p|}
\|\psi\|_{\M,\alpha,n}
.
\end{equation}
On the other hand, from (\ref{p0chi}) and inequality (\ref{p0 bound 6}) from
Lemma~\ref{lemma: mass shell geometry},
we get for any given $\epsilon>0$
\begin{equation}
\label{ineq p0 big}
|\chi(p)|
\le \frac{2m(m\vee \vec p)^{1-n}}{\frac{m}{12}\vee |p^0|}
e^{\alpha |\im \vec  p|}
\|\psi\|_{\M,\alpha,n}
\le \constl{3}(m\vee|\vec p|)^{-n}
e^{(\alpha+\epsilon) |\im \vec  p|}
\|\psi\|_{\M,\alpha,n}
\end{equation}
where $\constr{3}=\constr{3}(\epsilon, m):=
2m\constr{aux2}$.\\
{\it Case 2, ``locations close to the ramification set'':} $|p^0|< m/12$.
The key to deal with this case is provided by the following lemma,
which uses the geometric structure of $\M_\C$ close to the
ramification set.
The intuitive idea behind it relies on the fact that the 
three components of the differential form $d\vec p$ on $\M_\C$ 
become linearly dependent on the ramification set $Z$,
while the form $dp^0$ on $\M_\C$ does not vanish there. 
Consequently, close to any ramification point,
one can find a complex direction tangential to $\M_\C$
such that $p^0$ varies
considerably in that direction, while $d\vec p$ does not vary too much
in the same direction. This vague idea is made precise and quantitative
in the following lemma.
\begin{lemma}
\label{lemma: p0 nahe null}
For every $p=(p^0,\vec p)\in\M_{\C}$ with $|p^0|\le m/12$,
there is a holomorphic map
\begin{equation}
k=(k^0,\vec k):\bar{\Delta}\to \M_{\C}
\end{equation}
defined on the closed unit disc $\bar \Delta=\{t\in\C:\;|t|\le 1\}$
with 
\begin{align}
k(0)=p,&\\
\label{claim ktp}
|\vec{k}(t)-\vec p|\le m/6\quad&\text{ for }t\in\bar\Delta,\\
\label{claim k0p}
|k^0(t)|\ge m/12\quad &\text{ for }t\in\partial \Delta,
\end{align}
where $\partial \Delta=\{t\in\C:\;|t|=1\}$ denotes the unit circle.
\end{lemma}
This lemma is also proven in the appendix.
In the following estimates (\ref{eq:est1}) and (\ref{eq:est2}), 
we apply the function $k$ from this lemma 
together with the maximum principle for holomorphic functions.
Then, the inequalities (\ref{ineq p0 big-v1}) and (\ref{ineq p0 big}),
respectively,
are used with $p$ replaced by $k(t)$ with $t\in\partial\Delta$. 
The hypothesis $|k^0(t)|\ge m/12$ of these two inequalities is
verified by (\ref{claim k0p}).
Using also (\ref{claim ktp}) we get
\begin{align}
|\Phi(\vec p)|=|\chi(p)|&\le\sup_{t\in\partial\Delta}|\chi(k(t))|
\le 24
\sup_{t\in\partial\Delta}(m\vee |\vec k(t)|)^{-(n-1)}
e^{\alpha |\im \vec k(t)|}\|\psi\|_{\M,\alpha,n}
\cr&\le \constr{4-v1} (m\vee|\vec p|)^{-(n-1)}
e^{\alpha |\im \vec p|}
\|\psi\|_{\M,\alpha,n}
\label{eq:est1},
\\
\label{eq:est2}
|\Phi(\vec p)|=|\chi(p)|&\le
\constr{3}
\sup_{t\in\partial\Delta}(m\vee |\vec k(t)|)^{-n}
e^{(\alpha+\epsilon) |\im \vec k(t)|}\|\psi\|_{\M,\alpha,n}
\cr&
\le \constr{4} (m\vee|\vec p|)^{-n}
e^{(\alpha+\epsilon) |\im \vec p|}
\|\psi\|_{\M,\alpha,n}
\end{align}
with constants $\constr{4-v1}=\constr{4-v1}(n,\alpha,m)>24$
and $\constr{4}=\constr{4}(n,\alpha,\epsilon,m)>\constr{3}$.
This proves the bounds (\ref{bound F3M-v1}), (\ref{bound F3M}) and 
thus, the claim 
$\F_{3\M}\psi \in \CThree$, which yields 
$\F_{3\M}[\CM]\subseteq\CThree$.\\

\textsc{It remains to show the bound (\ref{bound FM3}) and
$\F_{\M3}[\CThree]\subseteq \CM$.} 
Given $\phi\in\CThree$, we have a holomorphic continuation
$\Phi:\C^3\to\C^4$ that fulfills the bound 
(\ref{bound 1})/(\ref{bound 1a}).
Then $\psi:=\F_{\M 3}\phi$ has the holomorphic continuation
\begin{equation}
\Psi:\M_\C\to\C^4,\quad
\Psi(p)=\frac{\slashed p+m}{2m}\gamma^{0}\Phi(\vec p)
\quad \text{ for }p=(p^0,\vec p)\in\M_\C.
\end{equation}
For the matrix norm $\|{\cdot}\|$ on $\C^{4\times 4}$
associated to the euclidean norm on $\C^4$, one observes
$\|\slashed p\|=|p|$ and $\|\gamma^0\|=1$.
The following estimate uses these two equalities in the first step,
the first inequality in (\ref{p0 bound 5}) in the second step, 
formula (\ref{bound 1a}) in the third step, 
and the second inequality in (\ref{p0 bound 5})
in the last step.
\begin{align}
|\Psi(p)|&\le \frac{|p|+m}{2m}|\Phi(\vec p)|
\le \frac{|p|}{m}|\Phi(\vec p)|
\le\frac{|p|}{m}|(m\vee |\vec p|)^{-n}e^{\alpha |\im \vec p|}
\|\phi\|_{3,\alpha,n}
\cr&\le \frac{3^{n/2}}{m}|p|^{-(n-1)}e^{\alpha |\im \vec p|}
\|\phi\|_{3,\alpha,n}.
\end{align}
Using the definition (\ref{bound 2a}) of  $\|\psi\|_{\M,\alpha, n}$,
this shows the bound (\ref{bound FM3}) and thus
$\F_{\M 3}\phi \in \CM$. \\

\textsc{We summarize:} We have proven $\F_{3\M}[\CM]\subseteq\CThree$
and $\F_{\M 3}[\CThree]\subseteq\CM$. Because  
$\F_{3\M}$ and $\F_{3\M}$ are inverse to each other,
this also proves the claims $\F_{3\M}[\CM]=\CThree$
and $\F_{\M 3}[\CThree]=\CM$.
\end{proof}

Finally, one observes the following corollary to Lemma~\ref{lem:FM3_F3M}.
\begin{corollary}\label{cor:CM_dense}
  $\CM$ is dense in $\HM$.  
\end{corollary}

\begin{proof}
  By Lemma~\ref{lem:CThree_dense}, $\CThree$ is dense in $\HThree$. Furthermore, by Lemma~\ref{lem:FM3_F3M}, the map $\F_{\M 3}:\HThree\to\HM$ is unitary and maps $\CThree$ onto $\CM$. The claim follows.
\end{proof}

\subsubsection{Properties of the Maps $\F_{\M\Sigma}$, $\F_{\Sigma\sol}$
and $\F_{\sol\M}$}
In this section, we prove three technical, but important lemmas.
The first one, Lemma~\ref{lemma: FMSigma}, deals with 
the generalized Fourier transformation $\F_{\M\Sigma}: \CSigma \to \CM$ 
from wave functions on $\Sigma$
to wave functions on the mass shell.
It relies on Paley-Wiener-like bounds: Support properties in
physical space are translated to growth rates in imaginary directions
in the complexified mass shell.
The second lemma, Lemma~\ref{lemma: FsolM},
deals with the maps $\F_{\sol\M}: \CM \to \Cs$
and $\F_{\Sigma\sol}: \Cs \to \CSigma$. \db Here, the point is \ed to translate
growth rates in imaginary directions in the complexified mass shell
back to support properties in physical space, using the classical 
Paley-Wiener theorem. Finally, the third lemma, 
Lemma~\ref{lemma: cycle Sigma-s-M},
is about compositions of these three maps. In particular,
it controls the support of a solution of the free Dirac equation
with given initial data on a Cauchy surface.

Recall the definitions of $\cal S(K,v)$, $\cal C_\Sigma(K)$, and
$\Vert{\cdot}\Vert_{\Sigma,K,n}$ given in the first paragraph of
Section \ref{sec:PWbounds}. The following lemma slightly extends Theorem~\ref{lemma:
FMSigma}.

\begin{lemma}[Paley-Wiener Bounds for Cauchy Surfaces]
    \label{lemma: FMSigma}
For any Cauchy surface $\Sigma$ the map $\F_{\M\Sigma}: \CSigma \to \CM$ is
well-defined. More precisely, let $K\subset\R^4$ be compact, $0\leq
V<1$,
$\Sigma\in\cal S(K,V)$,
$\psi\in\cal C_\Sigma(K)$,
$\alpha$ be a positive number such that $\alpha>\sqrt 2\sup_{x\in K}|x|$,
and $n\in\N$. Then
\begin{align}
    \label{eq:pw-bound}
    \|\F_{\M\Sigma}\psi\|_{\M,\alpha, n}
    \leq
    \constr{c:pwbound} \Vert \psi \Vert_{\Sigma,K,n}.
\end{align}
holds for some some positive constant
$\constr{c:pwbound}=\constr{c:pwbound}(K,V,n,\alpha,m)$. 
\end{lemma}

\begin{proof}
  The wave function $\psi\in\CSigma(K)$ is supported on the compact set $K\cap\Sigma$. 
We consider the following integral:
\begin{align}\label{eq:sigma-to-mc-integral}
 \Psi(p):=\frac{\slashed p+m}{2m}(2\pi)^{-3/2}\int_{K\cap\Sigma} e^{ipx}\,i_\gamma(d^4x)\,\psi(x)\qquad \text{for } p\in\M_{\C},
\end{align}
  which coincides with $(\F_{\M\Sigma}\psi)(p)$ for $p\in\M$; see
(\ref{def FMSigma}).
Because $(\slashed{p}-m)(\slashed{p}+m)=p^2-m^2=0$ holds for
$p\in\mathcal{M}_\C$, one has $(\slashed{p}-m)\Psi(p)=0$ for these $p$.
In particular, $(\F_{\M\Sigma}\psi)(p)\in \mathcal{D}_p$ 
holds for $p\in\M$. \\

\textsc{Our next goal is to estimate $\Psi(p)$ for $p\in\M_{\C}$.}
  For this we intend to use the decay of the Fourier transform 
for smooth compactly supported functions in $\R^3$,
at least for sufficiently large $|\re p|$. 
Therefore, we shall employ a projection $\pi_{\hat q}^\Sigma:\Sigma\to\R^{3}\simeq\Sigma^{0}=\{0\}\times \R^{3}$ in some time-like direction $\hat q\in\R^{4}$, $|\hat q|=1$, which fulfills:
\begin{enumerate}
 \item $\hat q$ is transversal to $\Sigma^{0}$ and to
     every tangent space of $\Sigma$.
 \item $\hat q\in (\re p)^{\perp}:=\left\{y\in\R^{4}\,|\,\re
     p_{\mu}y^{\mu}=0\right\}$ in order to have 
     \begin{align}
         \exp({i \re p_{\mu}
     x^{\mu}})
     =\exp({i \re p_{\mu}(s \hat q^{\mu} +x^{\mu})})
     \qquad \text{for any }s\in\R.
 \end{align}
\end{enumerate}
First, we focus on condition (a). Note that light-like vectors fulfill (a).  By
definition of $\cal S(K,V)$, the set
\begin{align}N=N(K,V):=\left\{\hat
q\in\R^{4}\,|\, |\hat q|=1, \hat q\text{ is transversal to $\Sigma^0$ and to all
$\Sigma\in\cal S(K,V)$}
\right\}\end{align} is a neighborhood of the set of light-like vectors in the unit
sphere. Therefore, we can choose $\epsilon=\epsilon(K,V)>0$ sufficiently small such
that \begin{align}E=E(K,V):=\left\{k\in\R^{4}\,\big|\,|k|=1,|k^{2}|\leq
\epsilon\right\}\subseteq N.\end{align} Note that $E$ is
compact. For $\Sigma\in\cal S(K,V),\hat
q \in E$, and $x\in \Sigma$ we define the projection
\begin{align}
    \left(0,\pi^\Sigma_{\hat q}(x)\right)
    =
    \left(s \hat{q}^{\mu} +x^{\mu}\right)_{\mu=0,1,2,3}\in\Sigma^{0}
\end{align} 
with $s=-x^0/\hat q^0$.  Note that $\pi^\Sigma_{\hat q}$ is a
diffeomorphism from $\Sigma$ to $\R^{3}$, and
$\pi^\Sigma_{\hat q}$ and all its derivatives depend continuously on $\hat q\in
E$.

Second, we focus on condition (b). To fulfill this condition the direction
$\hat q$ must be chosen to depend on $p\in\M_{\C}$, i.e., $\hat q=\hat q(p)$.
Therefore, for $p\in\C^{4}$ with $\re \vec p\neq 0$ we define
\begin{align}\hat{q}^{\mu}(p):=\frac{q^{\mu}(p)}{|q(p)|}\in\left(\re
    p\right)^{\perp} \qquad \text{with} \qquad q(p):=\left(|\re \vec p|,\left(\re
    p^{0}\right)\frac{\re\vec  p}{|\re\vec  p|}\right).
\end{align}

However, for this choice of $\hat q(p)$ it may occur that condition (a) is violated. Therefore we restrict $p$ to the following set
\begin{equation}
\label{eq:def-ISigma}
I=I(K,V):=\left\{p\in\M_{\C}\,\big|\, \left|\im p\right|^{2}+m^{2}\leq \epsilon\left|\re p\right|^{2}, \re \vec p\neq 0\right\}.
\end{equation}
Indeed, we have $\hat q[I]\subseteq E$ as the following argument shows. Let $p\in I$. We have $p^{2}-m^{2}=0$, and thus
\begin{align}
  \left(\re p\right)^{2}=\left(\im p\right)^{2}+m^{2}\leq \left|\im p\right|^{2}+m^{2}\leq \epsilon\left|\re p\right|^{2}.
\end{align}
In other words
\begin{align}\label{eq:re-p-hut}
  \frac{\re p}{|\re p|}\in E.
\end{align}
Furthermore, we have $q(p)^{2}=-(\re p)^{2}$ and 
$|q(p)|=|\re p|$, which imply
\begin{align}
  \frac{\re p}{|\re p|}\in E\Leftrightarrow \hat q(p)\in E.
\end{align}
Together with (\ref{eq:re-p-hut}) this shows $\hat q(p)\in E$.  In
consequence, $\hat q(p)$ fulfills conditions (a) and (b) for all $p\in
I$.\\

\textsc{In the next step, we provide a bound on $\Psi(p)$ defined in
(\ref{eq:sigma-to-mc-integral}) in the case of $p\in I$.} Using 
the transformation
$\vec y=\pi_{\hat q(p)}^\Sigma(x)$, which fulfills
$\re p x=-\re \vec p\cdot \vec y$ by construction,
we get
\begin{equation}
\label{eq:pi-transform}
  \int_{x\in K\cap\Sigma} e^{ipx}\,i_\gamma(d^4x)\,\psi(x)=\int_{\vec y\in
      \pi^\Sigma_{\hat
  q(p)}[K\cap\Sigma]} e^{-i \re \vec p \cdot \vec y}
  f_p(\vec y)\,d^{3}\vec y
\end{equation}
for
\begin{align}
  f_p(\vec y):=\exp\left(-\im p_{\mu}[(\pi^\Sigma_{\hat q(p)})^{-1}(\vec y)]^{\mu}\right)
  g_{\hat q(p)}(\vec y),
\end{align}
where
\begin{align}
  g_{\hat q}(\vec y)\,d^3\vec y
:=\left(( (\pi^\Sigma_{\hat q})^{-1}\right)^{*} \left(i_\gamma(d^4x)\,\psi(x) \right)
\end{align}
denotes the pull-back of $i_\gamma(d^4x)\,\psi(x)$ 
w.r.t.\ $(\pi^\Sigma_{\hat q})^{-1}$.

Thanks to compactness of $E$ and $K$ and continuity in $\hat
q\in E$, for all multi-indices $\beta\in \N_{0}^{3}$ with $|\beta|\leq n$,
the following holds, with 
the differentiation operators
$D^\beta$ acting on the variable $\vec y$:
\begin{align}
    \sup_{\Sigma\in\cal S(K,V)}\sup_{\hat q\in E}\sup_{\vec y\in\pi_{\hat
    q}[K\cap\Sigma]}\left|(D^{\beta}(\pi^\Sigma_{\hat q})^{-1})(\vec y)\right|<\infty,
\label{eq:supsup}
\end{align}
and hence,
\begin{equation}
\label{eq:bound_Dbeta_g}
\sup_{\Sigma\in\cal S(K,V)}\sup_{\hat q\in E}\sup_{\vec y\in\pi^\Sigma_{\hat q}[K\cap\Sigma]}
\left|D^{\beta}g_{\hat q}(\vec y)\right|<
\constl{c:q-bound}
\Vert\psi\Vert_{\Sigma,K,n}
,
\end{equation}
with a constant $\constr{c:q-bound}=\constr{c:q-bound}(K,V,n)$.
For any given $\delta>0$, taking
\begin{align}\label{eq:alpha}
  {\widetilde\alpha}=\widetilde \alpha(\delta,K):=2\delta+\sup_{x\in K}|x|,
\end{align}
we know 
\begin{align}
&\sup_{\Sigma\in\cal S(K,V)}\sup_{p\in I}\sup_{\vec y\in\pi^\Sigma_{\hat q(p)}[K\cap\Sigma]}
e^{(2\delta-{\widetilde\alpha})|\im p|}\exp\left(-\im p_{\mu}[(\pi_{\hat
q(p)})^{-1}(\vec y)]^{\mu}\right)\\
&=
\sup_{\Sigma\in\cal S(K,V)}\sup_{p\in I}\sup_{x\in K\cap\Sigma}e^{(2\delta-{\widetilde\alpha})|\im p|-x\im p}
\leq 1.
\end{align}
We obtain:
\begin{equation}
\label{eq: bound deriv}
\max_{|\beta|\le n}\,\sup_{\Sigma\in\cal S(K,V)}\sup_{p\in I}\sup_{\vec y\in\pi^\Sigma_{\hat
q(p)}[K\cap\Sigma]}
(m+|\im p|)^{-n}e^{(2\delta-{\widetilde\alpha})|\im p|}
\left|D^{\beta}\exp\left(-\im p_{\mu}[(\pi^\Sigma_{\hat q(p)})^{-1}(\vec y)]^{\mu}\right)\right|<\infty.
\end{equation}
To see this, one expresses the iterated derivatives $D^\beta\ldots$ with the
chain rule and uses the bound (\ref{eq:supsup}) for the inner derivatives
and compactness of the set
$\{(\hat q,\pi^\Sigma_{\hat q}(x))\;|\;\hat q\in E,\;x\in K\cap\Sigma, \Sigma\in\cal S(K,V)\}$.
Combining the bound (\ref{eq: bound deriv})
with the bound (\ref{eq:bound_Dbeta_g}) yields
\begin{align}
\max_{|\beta|\le n}\,\sup_{\Sigma\in\cal S(K,V)}\sup_{p\in I}\sup_{\vec y\in\pi^\Sigma_{\hat
q(p)}[K\cap\Sigma]}
(m+|\im p|)^{-n}e^{(2\delta-{\widetilde\alpha})|\im p|}\left|D^{\beta}f_p(\vec y)\right|
\leq
\constl{c:f-bound}\Vert\psi\Vert_{\Sigma,K,n},
\end{align}
and, a little weaker,
\begin{equation}
\label{eq:ingredient-bound-Dbeta}
\max_{|\beta|\le n}\,\sup_{\Sigma\in\cal S(K,V)}
\sup_{p\in I}\sup_{\vec y\in\pi^\Sigma_{\hat q(p)}[K\cap\Sigma]}
e^{(\delta-{\widetilde\alpha}) |\im p|}\left|D^{\beta}f_p(\vec y)\right|
\leq
\constl{c:f-bound-1}\Vert\psi\Vert_{\Sigma,K,n}
,
\end{equation}
for some constants $\constr{c:f-bound}=\constr{c:f-bound}(K,V,n,\delta,m)$ and 
$\constr{c:f-bound-1}=\constr{c:f-bound-1}(K,V,n,\delta,m)$.

Using compactness again, there is $R>0$ such that
\begin{equation}
\label{eq:ingredient-compactness}
  \bigcup_{p\in I}\operatorname{supp} f_p\subseteq 
  \bigcup_{\substack{q\in E\\\Sigma\in\cal S(K,V)}}
  \pi^\Sigma_{\hat q}[K\cap\Sigma]\subseteq \overline{B^3_R(0)}
,
\end{equation}
where $\overline{B^3_R(0)}$ denotes the closed $3$-dimensional ball with radius $R$ around $0$.

The bound (\ref{eq:ingredient-bound-Dbeta}) of the derivatives 
together with the boundedness (\ref{eq:ingredient-compactness}) 
yield the following bound for the Fourier transform
for some constant $\constl{const:Dbeta-bound}=
\constr{const:Dbeta-bound}(K,V,n,\delta,m)$.
\begin{align}
    \left|\widehat f_p(\vec k)\right|\leq 
\constr{const:Dbeta-bound}
\frac{e^{({\widetilde\alpha}-\delta) |\im p|}}
{(m+|\vec k|)^{n}}
\Vert\psi\Vert_{\Sigma,K,n}
,\qquad \forall\;\vec k\in\R^{3},
p\in I
.
\end{align}
In the special case $\vec k=\re\vec p$ it tells us the following for all
$p\in I$ and $n\in \N_0$:
\begin{align}
    \left|\int_{\vec y\in \pi^\Sigma_{\hat q(p)}[K\cap\Sigma]} e^{-i \re \vec p \cdot \vec y}
  f_p(\vec y)\,d^{3}\vec y\right|
  &\leq 
\constr{const:Dbeta-bound}
\,\frac{e^{({\widetilde\alpha}-\delta) |\im p|}}{(m+|\re\vec p|)^{n}}
\Vert\psi\Vert_{\Sigma,K,n}\\
&\le \constl{const:new-payley-wiener}
\frac{e^{{\widetilde\alpha} |\im p|}}{(m+|\vec p|)^{n}}
\Vert\psi\Vert_{\Sigma,K,n}
\le \constl{const:new-payley-wiener2}
|p|^{-n}e^{{\widetilde\alpha} |\im p|}
\Vert\psi\Vert_{\Sigma,K,n}
\end{align}
with some constants $\constr{const:new-payley-wiener}=
\constr{const:new-payley-wiener}(K,V,n,\delta,m)$ and 
$\constr{const:new-payley-wiener2}=
\constr{const:new-payley-wiener2}(K,V,n,m)$,
where in the last step we have used bound (\ref{p0 bound 5}) given in Lemma~\ref{lemma: mass shell geometry}.
Combining this with (\ref{eq:pi-transform}) and using
$\|\slashed p+m\|\le |p|+m\le 2|p|$
yields
\begin{align}
|\Psi(p)|=
\left|\frac{\slashed p+m}{2m}(2\pi)^{-3/2}\int_{K\cap\Sigma} e^{ipx}\,i_\gamma(d^4x)\,\psi(x)\right|
\le \constl{const:new-payley-wiener3}
|p|^{-(n-1)}e^{{\widetilde\alpha} |\im p|}
\Vert\psi\Vert_{\Sigma,K,n}
\qquad \text{for } p\in I
\label{eq:isigma-estimate}
\end{align}
with some \db sufficiently \ed large constant $\constr{const:new-payley-wiener3}=
\constr{const:new-payley-wiener3}(K,V,n,\delta,m)$ and
$\Psi$ being defined in (\ref{eq:sigma-to-mc-integral}).\\

\textsc{Next, we examine the easier case $p\in\M_\C\setminus I$.}
By definition (\ref{eq:def-ISigma}) of $I$, we have 
$\left|\im p\right|^{2}+m^{2}> \epsilon\left|\re p\right|^{2}$,
called ``case A'',
or $\re \vec p= 0$, called ``case B''. In case A, 
$|p|\le\constl{const:const-innen}(m+|\im p|)$ holds with some
constant $\constr{const:const-innen}=\constr{const:const-innen}
(\epsilon)$.
In case B, the same bound holds when $\constr{const:const-innen}$
is chosen sufficiently large. Indeed:
$\re \vec p= 0$ and $p^2=m^2$ imply
$|\re p|^2=(\re p^{0})^2=(\re p)^2=(\im p)^2+m^2\le|\im p|^2+m^2$,
and hence, $|p|^2\le 2|\im p|^2+m^2$.
Taking ${\widetilde\alpha}=2\delta+\sup_{x\in K}|x|$ as in (\ref{eq:alpha}) and 
a sufficiently large constant
$\constl{const:const-innen2}=
\constr{const:const-innen2}(K,V,n,\delta,m)$,
we conclude the following for $p\in \M_\C\setminus I$, $n\in \N$:
\begin{align}
\left|
\Psi(p)
\right|
&\le
\constr{const:const-innen2}|p|
e^{({\widetilde\alpha}-2\delta)|\im p|}
\Vert\psi\Vert_{\Sigma,K,n}
\le
\constr{const:const-innen2}|p|
e^{2\delta (m-|p|/\constr{const:const-innen})}
e^{{\widetilde\alpha}|\im p|}
\Vert\psi\Vert_{\Sigma,K,n}
\\
& \le
\constr{const:new-payley-wiener3}
|p|^{-(n-1)}e^{{\widetilde\alpha}|\im p|}
\Vert\psi\Vert_{\Sigma,K,n},
\label{eq: c11}
\end{align}
where the constant
$\constr{const:new-payley-wiener3}(K,V,n,\delta,m)$,
which was also used in (\ref{eq:isigma-estimate})
in a different way,
needs to be taken large enough. 
In the first inequality in (\ref{eq: c11}), we used
the definition (\ref{eq:sigma-to-mc-integral}) of $\Psi(p)$
and again the bound $\|\slashed p+m\|\le 2|p|$
together with the estimate 
$|e^{ipx}|\le e^{({\widetilde\alpha}-2\delta)|\im p|}$
for $p\in\M_\C\setminus I$ and $x\in K$.
Combining (\ref{eq: c11}) and (\ref{eq:isigma-estimate}) we have shown 
\begin{equation}
\left|\Psi(p)\right|
\le
\constr{const:new-payley-wiener3}
|p|^{-(n-1)}e^{{\widetilde\alpha}|\im p|}
\Vert\psi\Vert_{\Sigma,K,n}
\le
\constr{const:new-payley-wiener3}
|p|^{-(n-1)}e^{\sqrt{2}{\widetilde\alpha}|\im \vec p|}
\Vert\psi\Vert_{\Sigma,K,n}
\quad\text{ for all }p\in\mathcal{M}_\C,
\end{equation}
where we have used (\ref{p0 bound 3}) from Lemma \ref{lemma: mass shell geometry} in the last step.
Because $\Psi:
\M_\C\to\C^4$ is holomorphic, we can rewrite this as  
\begin{align}
    \|\F_{\M\Sigma}\psi\|_{\M,\sqrt{2}{\widetilde\alpha}, n}
    \leq
    \constr{const:new-payley-wiener3}\Vert\Psi\Vert_{\Sigma,K,n},
\end{align}
with the norm 
$\|{\cdot}\|_{\M,\sqrt{2}{\widetilde\alpha}, n}$ being defined in (\ref{bound 2a}). 
We now take a specific $\delta>0$ depending on the given $\alpha>\sqrt2
\sup_{x \in K} \Vert x \Vert$ such that the equation
\begin{align}
    \alpha=\sqrt 2\widetilde\alpha(\delta,K)
\end{align}
holds. This concludes the proof of
(\ref{eq:pw-bound}).
Hence, 
$\F_{\M\Sigma}: \CSigma \to \CM$ is well-defined which proves the claim.
\end{proof}

\begin{lemma}
\label{lemma: FsolM}
The maps $\F_{\sol\M}: \CM \to \Cs$
and $\F_{\Sigma\sol}: \Cs \to \CSigma$ are well-defined.
For any $\psi\in\CM$ and $\alpha>0$ such that 
$\|\psi\|_{\M,\alpha, n}<\infty$ holds for all $n\in\N$,
the function $\F_{\sol\M}\psi$ is supported in
$\{0\}\times \overline{B^3_{\alpha}(0)}+\causal$.
\end{lemma}
\begin{proof}
\textsc{\db First, we show \ed that $\F_{\sol\M}: \CM \to \Cs$ is well-defined.}
Let $\psi\in \CM$, $\Psi:\M_\C\to\C^4$ be the holomorphic
extension of $\psi$ to $\mathcal{M}_\C$, and take $\alpha>0$ such that
for all $n\in\N$ the bound 
$\|\psi\|_{\M,\alpha, n}<\infty$ holds.
Because $|\psi(p)|$ tends to $0$ as $|p|\to\infty$, $p\in\mathcal{M}$,
faster than any power of $|p|$, 
\begin{align}
\F_{\sol\M}\psi(x)=
\frac{(2\pi)^{-3/2}}{m}\int_\M e^{-i px}\psi(p)\,i_{p}(d^{4}p)
\end{align}
depends smoothly on $x\in\R^4$.
Furthermore, $\F_{\sol\M}\psi$ solves the free Dirac equation
because the fact $\psi(p)\in\mathcal{D}_p$ for any $p\in\mathcal{M}$
implies
\begin{align}
(i\slashed{\partial}-m)\F_{\sol\M}\psi(x)
=
\frac{(2\pi)^{-3/2}}{m}\int_\M e^{-i px}(\slashed{p}-m)
\psi(p)\,i_{p}(d^{4}p)=0\quad \text{ for }x\in\R^4.
\end{align}
Given $t\in\R$, we introduce the time-shifted version
\begin{align}
\Psi_t(p):=e^{-ip_0t}\Psi(p),\qquad p\in\mathcal{M}_\C.
\end{align}
The restriction of this holomorphic map to $\mathcal{M}$ is denoted by 
$\psi_t$. We observe the following
for $p=(p^0,\vec p)\in\mathcal{M}_\C$, using the bound
$|\im p^0|\le |\im \vec p|$ from (\ref{p0 bound 2}):
\begin{align}
|\Psi_t(p)|\le e^{|\im p_0||t|}|\Psi(p)|
\le
e^{|\im \vec p||t|}|\Psi(p)|
\end{align}
We get
\begin{align}
\|\psi_t\|_{\M,\alpha+|t|, n}
\le
\|\psi\|_{\M,\alpha, n},\qquad (n\in\N);
\end{align}
recall the definition of $\|{\cdot}\|_{\M,\alpha, n}$
\db from (\ref{bound 2a})\ed.
Then the bound (\ref{bound F3M-v1}) implies
\begin{align}
\|\F_{3\M}\psi_t\|_{3,\alpha+|t|, n-1}
\le 
\constr{4-v1}\|\psi_t\|_{\M,\alpha+|t|, n}
\le\constr{4-v1}\|\psi\|_{\M,\alpha,n}<\infty,\qquad (n\in\N).
\end{align}
Using this, 
the classical Paley-Wiener Theorem \cite[Theorem IX.11]{reed_methods_1981}
implies that the inverse Fourier transform
\begin{align}
\R^3\ni\vec x\mapsto 
(2\pi)^{-3/2}\int_{\R^3} e^{i\vec p\vec x}\F_{3\M}\psi_t(\vec p)\,
d^3\vec p
\end{align}
is supported in the ball $\overline{B^3_{\alpha+|t|}(0)}$.
Taking $x=(t,\vec x)$ we compute
\begin{align}
&(2\pi)^{-3/2}\int_{\R^3} e^{i\vec p\vec x}\F_{3\M}\psi_t(\vec p)\, d^3\vec p
=
\frac{(2\pi)^{-3/2}}{m}\int_{\R^3} e^{i\vec p\vec x} \left[\psi_{t}(p_{+}(\vec
p))-\psi_{t}(p_{-}(\vec p))\right]\frac{m^{2} d^3\vec p}{E(\vec p)}\\ 
&=
\frac{(2\pi)^{-3/2}}{m}\int_{\R^3} \left[e^{-i p_{+}(\vec p) x}\psi(p_{+}(\vec
p))-e^{-i p_{-}(\vec p) x}\psi(p_{-}(\vec p))\right]\frac{m^{2} d^3\vec
p}{E(\vec p)}\\ 
&= \frac{(2\pi)^{-3/2}}{m}\int_\M e^{-i p
x}\psi(p)\,i_{p}(d^{4}p)\qquad\text{by (\ref{eq:ipd4p})},\\
&=
\F_{\sol\M}\psi(x).
\end{align}
This shows that $\F_{\sol\M}\psi$ is supported in $\{0\}\times
\overline{B^3_{\alpha}(0)}+\causal$.  Consequently, $\F_{\sol\M}$ maps $\CM$ to
$\Cs$.\\

\textsc{Finally, we consider $\F_{\Sigma\sol}$.} Let $\psi\in\Cs$ be supported
in $K+\causal$ with some compact set $K\subset\R^4$.  Because $\psi$ is smooth,
its restriction to $\Sigma$ is also smooth.  Moreover, $(K+\causal)\cap\Sigma$
is compact.  This shows that $\F_{\Sigma\sol}\psi\in\CSigma$.
\end{proof}

\begin{lemma} 
\label{lemma: cycle Sigma-s-M}
For $\psi\in\CSigma$,
the function $\F_{\sol\M} \F_{\M\Sigma}\psi$
is supported in $\operatorname{supp}\psi+\causal$.
Furthermore, the following identities hold:
\begin{align}
  \F_{\Sigma\sol}  \F_{\sol\M}  \F_{\M\Sigma}
  &=\mathrm{id}_{\CSigma},\label{eq:first_id}\\
  \F_{\M\Sigma}  \F_{\Sigma\sol}  \F_{\sol\M}
  &=\mathrm{id}_{\CM},\label{eq:second_id} \\
  \F_{\sol\M}  \F_{\M\Sigma} \F_{\Sigma\sol}
  &=\mathrm{id}_{\Cs}.\label{eq:third_id}
\end{align}
Finally, the maps $\F_{\M\Sigma}:\CSigma\to\CM$, 
$\F_{\Sigma\sol}:\Cs\to\CSigma$, 
and $\F_{\sol\M}:\CM\to\Cs$ are isometric isomorphisms.
\end{lemma}

\begin{proof}  
We abbreviate
\begin{align}
\F_{\sol\Sigma}:=
\F_{\sol\M} \F_{\M\Sigma}.
\end{align}
\textsc{We examine first 
the support of $\F_{\sol\Sigma}$ and of
$\F_{\Sigma\sol}  \F_{\sol\Sigma}\psi$
for $\psi\in\CSigma$.}
We claim: If the support of $\psi$ is contained
in $\overline{B^4_r(y)}$ for some $y\in\Sigma$, then the support of
$\F_{\sol\Sigma}\psi$
is contained in 
$y+\{0\}\times\overline{B^3_{\sqrt{2}r}(0)}+\causal$.
We prove this first for the special case $y=0\in\Sigma$.
Under the above assumptions, Lemma~\ref{lemma: FMSigma} yields 
$\|\F_{\M\Sigma}\psi\|_{\M,\sqrt{2}\alpha, n}<\infty$ 
for all $\alpha>r$ and all $n\in\N$. 
Using  Lemma~\ref{lemma: FsolM}, it follows that the support of
$\F_{\sol\Sigma}(\psi)$
is contained in $\{0\}\times \overline{B^3_{\sqrt{2}\alpha}(y)}+\causal$.
Because $\alpha>r$ is arbitrary, this implies that
$\operatorname{supp}(\F_{\sol\Sigma}\psi)
\subseteq \{0\}\times \overline{B^3_{\sqrt{2}r}(y)}+\causal$.
Next, we reduce the general case $y\in\Sigma$
to the special case $y=0$, using the translation maps
$T_\Sigma^{-y}:\CSigma\to\cC_{\Sigma-y}$ and
$T_{\sol}^{-y}:\Cs\to\Cs$
from Definition \ref{def: Poincar'e transformations}.
By equations (\ref{eq:translation-M-Sigma}) and (\ref{eq:translation-s-M}) 
in Lemma~\ref{thm: compatibility Poincar'e}, these maps fulfill
\begin{align}
T_{\sol}^{-y} \F_{\sol\Sigma}
=\F_{\sol,\Sigma-y}  T_{\Sigma}^{-y}.
\end{align}
Given that $\psi$ is supported in 
a subset of $\overline{B^4_r(y)}$, it follows that
$T_{\Sigma}^{-y}\psi$ 
is supported in 
a subset of $\overline{B^4_r(0)}$. Using the special case from above,
it follows that 
$T_{\sol}^{-y} \F_{\sol\Sigma}\psi
=\F_{\sol,\Sigma-y}  T_{\Sigma}^{-y}\psi$
is supported in 
$\{0\}\times\overline{B^3_{\sqrt{2}r}(0)}+\causal$.
But then $\F_{\sol\Sigma}\psi$
is supported in
$y+\{0\}\times\overline{B^3_{\sqrt{2}r}(0)}+\causal$.\\

\textsc{We prove now the first claim of the lemma.} 
Let $r>0$. Using a compactness
argument and a partition of unity, we can take finitely many points 
$y_1,\ldots, y_k\in\operatorname{supp}\psi$
and $\psi_1,\ldots,\psi_k\in\CSigma$ with $\sum_{j=1}^k\psi_j=\psi$,
such that for $j=1,\ldots,k$, we have $\supp\psi_j\subseteq B^4_r(y_j)$.
But then 
\begin{align}
\supp(\F_{\sol\Sigma}\psi_j)
\subseteq
y_j+\{0\}\times\overline{B^3_{\sqrt{2}r}(0)}+\causal,
\quad
j=1,\ldots,k.
\end{align}
We conclude
\begin{align}
\supp(\F_{\sol\Sigma}\psi)
\subseteq
\bigcup_{j=1}^k\left(
y_j+\{0\}\times\overline{B^3_{\sqrt{2}r}(0)}+\causal
\right)
\subseteq
\supp\psi +\{0\}\times\overline{B^3_{\sqrt{2}r}(0)}+\causal.
\end{align}
Because $r>0$ is arbitrary, this proves the claim
\begin{align}
\supp(\F_{\sol\Sigma}\psi)
\subseteq
\bigcap_{r>0}
\left(
\supp\psi +\{0\}\times\overline{B^3_{\sqrt{2}r}(0)}+\causal
\right)
=
\supp\psi +\causal.
\end{align}
We get for any $\psi\in\CSigma$,
using that $\supp\psi\subseteq \Sigma$ is space-like:
\begin{equation}
\label{eq: support stabil}
\supp(\F_{\Sigma\sol}
 \F_{\sol\Sigma}\psi)
\subseteq 
\Sigma\cap(\supp\psi +\causal)
=\supp\psi
.
\end{equation}
\textsc{Next, we prove equation~(\ref{eq:first_id}):}
Given $\psi\in\CSigma$ and $y\in\Sigma$,
we need to show
\begin{equation}
\label{eq:claim-initial-cond-fulfilled}
\F_{\Sigma\sol}
 \F_{\sol\Sigma}\psi(y)=\psi(y).
\end{equation}
We prove this first in the special case that $y=0\in\Sigma$
and that the tangent space of $\Sigma$ in $0$ equals
$T_0\Sigma=\{0\}\times\R^3$,
and then reduce the general case to the special case, 
using a translation and a Lorentz transformation.

Let 
$\R^3\ni\vec x\mapsto (t_\Sigma(\vec x),\vec x)\in\Sigma$
be the representation of $\Sigma$ as a graph as in 
(\ref{eq:parametrize Sigma}).
In particular our assumption means $t_\Sigma(0)=0$ and $\nabla t_\Sigma(0)=0$.
We set $\eta=\sup_{\vec x\in\R^3}|\nabla t_\Sigma(\vec x)|<1$;
recall condition (c) in the definition of Cauchy surfaces
(Def. \ref{def:cauchy-surface}).
Then for every $x=(x^0,\vec x)\in\Sigma\setminus\{0\}$ and every 
$y=(0,\vec y)\in\{0\}\times\R^3$ with $|\vec y|<(1-\eta)|\vec x|$,
the vector $x-y$ is space-like.
Indeed, $|x^0-y^0|=|x^0|\le \eta|\vec x|<|\vec x|-|\vec y|\le |\vec x-\vec y|$.
Let $\phi:\R^3\to\R_0^+$ be a smooth function supported in
the open ball $B_{1-\eta}^3(0)$ with 
\begin{equation}
\label{eq: integral 1}
\int_{\R^3} \phi(\vec x)\,d^3\vec x=1,
\end{equation}
and let $\chi:\R^3\to[0,1]$ be another smooth, compactly supported function
which equals $1$ in $\overline{B_1^3(0)}$.
For every $\epsilon>0$, we introduce $\chi_\epsilon:\Sigma\to[0,1]$,
$\chi_\epsilon(x^0,\vec x)=\chi(\vec x/\epsilon)$
and $\phi_\epsilon:\R^3\to\R_0^+$, 
$\phi_\epsilon(\vec x)=\epsilon^{-3}\phi(\vec x/\epsilon)$.
Note that $\phi_\epsilon$ fulfills 
\begin{equation}
\label{eq:integral-phi-skaliert}
\int_{\R^3} \phi_\epsilon(\vec x)\,d^3\vec x=1.
\end{equation}
Furthermore, for every $x=(x^0,\vec x)\in \supp((1-\chi_\epsilon)\psi)$, we have
$|\vec x|\ge \epsilon$, and every $y=(0,\vec y)\in 
\{0\}\times \supp\phi_\epsilon$
fulfills $|\vec y|<(1-\eta)\epsilon$. Hence
$x-y$ is space-like. It follows
that the sets $\supp((1-\chi_\epsilon)\psi)+\causal$
and $\{0\}\times \supp\phi_\epsilon$ are disjoint.
Using 
\begin{align}
\supp(\F_{\sol\M} \F_{\M\Sigma}
((1-\chi_\epsilon)\psi)
\subseteq 
\supp((1-\chi_\epsilon)\psi)+\causal,
\end{align}
we conclude for all $\vec x\in\R^3$, $x=(0,\vec x)$ and $\epsilon>0$
that $\phi_\epsilon(\vec x)=0$ or $\F_{\sol\Sigma}((1-\chi_\epsilon)\psi)(x)=0$
holds, i.e.,
\begin{align}
\label{eq:cutoff-chi-ok}
\phi_\epsilon(\vec x)\,\F_{\sol\Sigma}\psi(x)=
\phi_\epsilon(\vec x)\,\F_{\sol\Sigma}
(\chi_\epsilon\psi)(x).
\end{align}
Integrating the left hand side 
over $\vec x$ and taking the limit as $\epsilon\downarrow 0$ we
get on the one hand, using continuity of 
the function $\F_{\sol\Sigma}\psi$:
\begin{align}\label{eq:lim-cutoff-chi-ok}
\lim_{\epsilon\downarrow 0}\int_{\R^3}
\phi_\epsilon(\vec x)\cdot\F_{\sol\Sigma}\psi(0,\vec x)
\,d^3 \vec x=\F_{\sol\Sigma}\psi(0).
\end{align}
On the other hand, we integrate also
the right hand side of (\ref{eq:cutoff-chi-ok}) and rewrite it as
\begin{align}
\nonumber
&\int\limits_{\R^3}\phi_\epsilon(\vec x)
[\F_{\sol\Sigma}(\chi_\epsilon\psi)](0,\vec x)\,d^3\vec x
\\
\nonumber
&=
\frac{(2\pi)^{-3/2}}{m}\epsilon^{-3}
\int\limits_{\vec x\in\R^3}\phi(\vec x/\epsilon)
\int\limits_{p=(p^0,\vec p)\in\M}e^{i\vec p\vec x} 
{\F}_{\M\Sigma}(\chi_\epsilon\psi)(p)\,i_p(d^4p)\,d^3\vec x
\\
\nonumber
&=
\frac{(2\pi)^{-3/2}}{m}\epsilon^{-3}
\int\limits_\M
\int\limits_{\R^3}\phi(\vec x/\epsilon)
e^{i\vec p\vec x} \,d^3\vec x\,
{\F}_{\M\Sigma}(\chi_\epsilon\psi)(p)\,i_p(d^4p)
\\&=
m^{-1}
\int\limits_\M
\hat{\phi}(\epsilon\vec p)
{\F}_{\M\Sigma}
(\chi_\epsilon\psi)(p)\,i_p(d^4p)
\label{eq:change-order-integration}
\end{align}
with the Fourier integral \db being \ed
\begin{align}
\hat{\phi}(\vec q):=
(2\pi)^{-3/2}
\int\limits_{\R^3}\phi(\vec x)
e^{i\vec q\vec x} \,d^3\vec x,\quad \vec q\in\R^3.
\end{align}
Note that changing the order of integration
in (\ref{eq:change-order-integration}) is justified because
$\phi$ is compactly supported and because
${\F}_{\M\Sigma}
(\chi_\epsilon\psi)(p)$ decays faster than any power of $|p|$
as $|p|\to\infty$, $p\in\M$;
here we use that ${\F}_{\M\Sigma}
(\chi_\epsilon\psi)\in\CM$ by Lemma~\ref{lemma: FMSigma}.
Using the definition~(\ref{def FMSigma}) of ${\F}_{\M\Sigma}$,
formulas (\ref{eq:Pp})--(\ref{eq: P+P-}) and (\ref{eq:ipd4p}), 
and the representation (\ref{eq: repr igammad4x}) of $i_\gamma(d^4x)$,
the quantity in (\ref{eq:change-order-integration})
equals
\begin{align}
\nonumber
&\frac{(2\pi)^{-3/2}}{m}
\int_{p\in\M}
\hat{\phi}(\epsilon\vec p)
\frac{\slashed{p}+m}{2m}
\int_{x\in\Sigma}e^{ipx}\,
i_\gamma(d^4x)\,
\chi_\epsilon(x)\psi(x)\,i_p(d^4p)
\\\nonumber
=&
(2\pi)^{-3/2}\int_{p\in\M}
\hat{\phi}(\epsilon\vec p)
P(p)
\int_{x\in\Sigma}e^{ipx}\,
\gamma^0 i_\gamma(d^4x)\,
\chi_\epsilon(x)\psi(x)\,d^3p
\\\nonumber
=&
(2\pi)^{-3/2}\int_{\vec p\in\R^3}
\hat{\phi}(\epsilon\vec p)
\int_{\vec x\in\R^3}
\left(
P_+(\vec p)e^{iE(\vec p) t_\Sigma(\vec x)}
+P_-(\vec p)e^{-iE(\vec p) t_\Sigma(\vec x)}
\right)
e^{-i\vec p\vec x} \times
\\\nonumber
&\times
\left(1+\sum_{\mu=1}^3\gamma^0\gamma^\mu 
{\partial_\mu t_\Sigma(\vec x)}\right)
\chi(\vec x/\epsilon)\psi(t_\Sigma(\vec x),\vec x)\,d^3\vec x\,d^3\vec p
\\\nonumber
=&
(2\pi)^{-3/2}\int_{\vec q\in\R^3}
\hat{\phi}(\vec q)
\int_{\vec y\in\R^3}
\left(
P_+(\vec q/\epsilon)e^{iE(\vec q/\epsilon) t_\Sigma(\epsilon\vec y)}
+P_-(\vec q/\epsilon)e^{-iE(\vec q/\epsilon) t_\Sigma(\epsilon\vec y)}
\right)
e^{-i\vec q\vec y}\times\\
&\times\left(1+\sum_{\mu=1}^3\gamma^0\gamma^\mu 
{\partial_\mu t_\Sigma(\epsilon\vec y)}\right)
\chi(\vec y)\psi(t_\Sigma(\epsilon\vec y),\epsilon\vec y)
\,d^3\vec y\,d^3\vec q \label{eq:ini-val-epsilon}
\end{align}
We now take the limit as $\epsilon$ tends to zero using dominated convergence
and \db exploit \ed the following ingredients:
\begin{enumerate}[(a)]
  \item To find a dominating function for the integrand we employ:
    \begin{enumerate}[(i)]
      \item $P_{+}$ and $P_{-}$ take values in the set of orthogonal projectors and therefore are bounded;
      \item $\nabla t_\Sigma$ and $\psi$ are bounded;
      \item $\hat \phi(\vec q)$ is bounded and decays faster than any power of $|\vec q|$ for $|\vec q|\to\infty$;
      \item $\chi$ is bounded and compactly supported;
    \end{enumerate}
  \item For the point-wise convergence as of the integrand $\epsilon\to 0$ we use for any $\vec q,\vec x$:
    \begin{enumerate}[(i)]
      \item $P_{\pm}(\vec q/\epsilon)$ converge to orthogonal projectors and $P_{+}(\vec q/\epsilon)+P_{-}(\vec q/\epsilon)$ equals the identity;
      \item $t_\Sigma(\epsilon\vec y)/\epsilon\to 0$ and $E(\vec q/\epsilon)\epsilon$ is bounded for sufficiently small $\epsilon$;
      \item $\nabla t_\Sigma(\epsilon \vec y)\to 0$;
      \item $\psi(t_\Sigma(\epsilon\vec y),\epsilon\vec y)\to \psi(0)$.
    \end{enumerate}
\end{enumerate}
This implies that the limit of (\ref{eq:ini-val-epsilon}) 
as $\epsilon\to 0$ can be expressed as
\begin{align}
&  \psi(0)(2\pi)^{-3/2}\int_{\vec q\in\R^3}
\hat{\phi}(\vec q)
\int_{\vec y\in\R^3}
e^{-i\vec q\vec y}
\chi(\vec y)\,d^3\vec y\,d^3\vec q
\cr
=&
  \psi(0)\int_{\vec y\in\R^3}
\phi(\vec y)
\chi(\vec y)\,d^3\vec y
=
  \psi(0)\int_{\vec y\in\R^3}
\phi(\vec y)
\,d^3\vec y
\cr
=&
  \psi(0)
\end{align}
because by the choice of $\phi$ and $\chi$ we have $\phi\chi=\phi$;
recall also formula (\ref{eq: integral 1}). 
Let us summarize. Together with (\ref{eq:cutoff-chi-ok}), 
(\ref{eq:lim-cutoff-chi-ok}) we have shown that
\begin{align}
  \F_{\sol\Sigma}\psi(0)=\psi(0)
\end{align}
which implies
\begin{align}
  \F_{\Sigma\sol}  \F_{\sol\Sigma}\psi(0)=\psi(0).
\end{align}
Next, we treat the case of general $y\in\Sigma$ with a general
tangent space $T_y\Sigma$, using a translation by $-y$
and a Lorentz transformation encoded by some $(S,\Lambda)$
that maps the space-like hyperplane $T_0(\Sigma-y)$ 
to the time-0-hyperplane $\{0\}\times\R^3$.
Using Theorem~\ref{thm: compatibility Poincar'e}
together with the special case just considered, we get 
\begin{align}
  &\F_{\Sigma\sol}\F_{\sol\Sigma}\psi(y)=\F_{\sol\Sigma}\psi(y)=
T_{\sol}^{-y} \F_{\sol\Sigma}\psi(0)=
\F_{\sol,\Sigma-y}T_\Sigma^{-y}\psi(0)
\\&
=
(L_{\sol}^{(S,\Lambda)})^{-1}\F_{\sol,\Lambda(\Sigma-y)} 
L_{\Sigma-y}^{(S,\Lambda)}T_\Sigma^{-y}\psi(0)
=
S^{-1}\cdot L_{\Sigma-y}^{(S,\Lambda)}T_\Sigma^{-y}\psi(0)
=T_\Sigma^{-y}\psi(0)
=\psi(y).
\end{align}
This proves equation~(\ref{eq:first_id}).

By Definition (\ref{eq:scalar-product-sol}) of the scalar product in $\Cs$,
the map $\F_{\Sigma\sol}:\Cs\to\CSigma$ is an isometry.
Using equation~(\ref{eq:first_id}), i.e.,
$\F_{\Sigma\sol}  \F_{\sol\Sigma}=\operatorname{id}_{\CSigma}$,
it follows that $\F_{\sol\Sigma}$ is also an isometry.\\

\textsc{\db Now we prove \ed that $\F_{\sol\M}:\CM\to\Cs$ is
an isometry.}
We consider the time-zero-hyperplane $\Sigma^0=\{0\}\times \R^3$
and set \db $\F_{\Sigma^0 3}:=
\F_{\Sigma^0 \sol}  \F_{\sol\M}  \F_{\M 3}$.\ed
It is just the standard inverse Fourier transform,
as the following calculation shows.
For $\psi\in\CThree$ and $\vec x\in\R^3$, combining
(\ref{def FSigmas}), (\ref{def FsM}), and (\ref{def FM3})
with
(\ref{eq:Pp}), (\ref{eq: P+-}), (\ref{eq: P+P-}), and (\ref{eq:ipd4p}),
we obtain
\begin{align}
\db (\F_{\Sigma^0 3}\ed\psi)(0,\vec x)
=(2\pi)^{-3/2}\int_{\R^3}e^{i\vec p\vec x}(P_+(\vec p)+P_-(\vec p))
\psi(\vec p)\,d^3\vec p
 =(2\pi)^{-3/2}\int_{\R^3}e^{i\vec p\vec x}\psi(\vec p)\,d^3\vec p
\end{align}
As a consequence, \db $\F_{\Sigma^03}\psi:\CThree\to\cC_{\Sigma^0}$ \ed is isometric. Since
\db $\F_{\Sigma^0\sol}$ \ed and $\F_{\M 3}$ are isometries
and $\F_{\M 3}[\CThree]=\CM$ by Lemma~\ref{lem:FM3_F3M}, it follows that 
$\F_{\sol\M}:\CM\to\Cs$ is also an isometry. 

As $\F_{\Sigma\sol}:\Cs\to\CSigma$ and $\F_{\sol\M}:\CM\to\Cs$ are isometric,
formula (\ref{eq:first_id}) implies that $\F_{\M\Sigma}:\CSigma\to\CM$ is also
isometric.\\
  
\textsc{Finally, we prove equations (\ref{eq:second_id}) 
and (\ref{eq:third_id}).} 
We use the following well-known fact. 
Assume that isometries $f:C\to C'$ and $g:C'\to C$ between 
pre-Hilbert spaces $C,C'$ are given. Further assume that 
$g\circ f=\mathrm{id}_{C}$ holds. 
Then $f$ and $g$ are isometric isomorphisms and inverse to each other. 
We apply this fact to $g=\F_{\Sigma\sol}  \F_{\sol\M}$ and 
$f=\F_{\M\Sigma}$ on the one hand to get 
equation~(\ref{eq:second_id}) from equation~(\ref{eq:first_id}), 
and to $g=\F_{\Sigma\sol}$ and 
$f=\F_{\sol\M}  \F_{\M\Sigma}$ on the other hand 
to get equation~(\ref{eq:third_id}) also from equation~(\ref{eq:first_id}).
The three equations  (\ref{eq:first_id})--(\ref{eq:third_id})
show also that the three maps $\F_{\M\Sigma}: \CSigma \to \CM$, 
$\F_{\Sigma\sol}: \Cs \to \CSigma$, and $\F_{\sol\M}: \CM \to \Cs$
are isomorphisms. 
\end{proof}

As a consequence we get the following corollary.

\begin{corollary}\label{cor:FMSigma_properties}
The maps $\F_{\M\Sigma}: \CSigma \to \CM$, $\F_{\Sigma\sol}: \Cs \to \CSigma$, and $\F_{\sol\M}: \CM \to \Cs$ extend to unitary maps
\begin{align}\F_{\M\Sigma}: \HSigma \to \HM,\qquad \F_{\Sigma\sol}: \Hs \to \HSigma,\qquad \F_{\sol\M}: \HM \to \Hs.
\end{align} Furthermore, they fulfill
\begin{align}
  \F_{\Sigma\sol}  \F_{\sol\M}  \F_{\M\Sigma}&=\mathrm{id}_{\HSigma},\\
  \F_{\M\Sigma}  \F_{\Sigma\sol}  \F_{\sol\M}&=\mathrm{id}_{\HM}, \\
  \F_{\sol\M}  \F_{\M\Sigma} \F_{\Sigma\sol}&=\mathrm{id}_{\Hs}.
\end{align}
\end{corollary}

\begin{proof}
This follows immediately from Lemma~\ref{lemma: cycle Sigma-s-M},
because
$\CSigma$, $\Cs$, and $\CM$ are dense in $\HSigma$, $\Hs$, and $\HM$,
respectively.
\end{proof}

\subsubsection{Proof of Theorem
  \ref{thm:F_maps}}\label{sec:generalized-fourier}

\begin{proof}[Proof of Theorem~\ref{thm:F_maps}]
(a) For any placeholders $I,J,K$ among 
 the symbols $3,\M,\sol$ or any Cauchy surface $\Sigma$
such that $\F_{IJ}$ and $\F_{JK}$
are already defined, but $\F_{IK}$ is not yet defined,
we define $\F_{IK}:=\F_{IJ} \F_{JK}$.
This is repeated recursively until all maps $\F_{IK}$ are defined.
As a consequence of Lemmas \ref{lem:FM3_F3M} and 
\ref{lemma: cycle Sigma-s-M} and Corollary
\ref{cor:FMSigma_properties}, these recursive 
definitions do not contradict each
other. All claims of Theorem~\ref{thm:F_maps} (a) follow now immediately.
Claim (b) is already proven in Lemma~\ref{lemma: cycle Sigma-s-M},
and claim (c) is just composed of Lemma~\ref{lem:CThree_dense}, 
Corollary \ref{cor:CM_dense}, Definition~\ref{def: CSigma}, 
and equation~(\ref{eq:Cs_dense}).
\end{proof}

%\subsection{External Potentials}\label{sec:dirac-extern}
%
%\db This last part is devoted to the solution theory of the
%one-particle Dirac equation
%\begin{align}\tag{\ref{eq:dirac-equation}}
%  (i\slashed\partial-\slashed A) \psi=m\psi
%\end{align}
%for a smooth and compactly supported potential 
%\begin{align}\tag{\ref{def-A}}
%  A=(A_{\mu})_{\mu=0,1,2,3}\in\cC^{\infty}_{c}(\mathbb R^{4},\mathbb R^{4}).
%\end{align}
%It is arranged as follows: In Section~\ref{sec:sobolev} we introduce convenient
%norms to study regularity.  The proof of
%Theorem~\ref{dirac-existence-uniqueness} is split up in several lemmas over
%which Section~\ref{sec:dirac-et-overview} provides an overview.
%Theorem~\ref{thm: Hilbert space time evolution} is proven in 
%Section~\ref{sec:evolution-hilbert}. It is a direct consequence of
%Theorem~\ref{dirac-existence-uniqueness}.
%
%\ed 

\subsection{Existence, Uniqueness, and Causal Structure}
\label{sec:dirac-et-overview}

In this section the Theorems~\ref{dirac-existence-uniqueness}, \ref{thm: Hilbert space
time evolution}, and \ref{lemma: equiv} are proven.
The strategy of proof is the following:
\begin{enumerate}[(1)]
    \item Proof of existence and uniqueness of solutions in the interaction
        picture introduced in Section~\ref{sec:interaction};
        see Lemma~\ref{lem:interaction-existence-uniqueness}.
    \item Proof of regularity and support properties of solutions in the
        interaction picture; see Theorem~\ref{lemma: equiv}.
    \item Use (2) to prove the equivalence of the Schr\"odinger and the interaction
        picture in Section~\ref{subsect:Regularity and support properties}.
    \item Use (1), (2), and (3) to prove existence, uniqueness, regularity, and causal
        structure of solutions for the Dirac equation in
        Sections \ref{sec:main-result-1} and \ref{sec:evolution-hilbert}. 
\end{enumerate}

\subsubsection{\db Existence and Uniqueness in the Interaction Picture\ed}
\label{subsect:Existence and Uniqueness in the interaction picture}

Preliminarily we check general properties of the operator $L_t$ defined in 
(\ref{eq:interaction-evolution-Cs}).

\begin{lemma}\label{lem:reg-Lt}
For any $t\in\R$ and $n\in\N_0$, the operator $L_t:\Cs\selfmaps$
introduced in (\ref{eq:interaction-evolution-Cs})
extends to a bounded linear map 
$L_t:(\cH_{\sol,n},\|{\cdot}\|_{\sol,n})\selfmaps$, 
denoted by the same symbol $L_t$.
For any $n,l\in\N_0$ with $n\geq l$ the function
\begin{equation}
    L_{(\cdot)}:\R\to\mathcal{B}(\cH_{\sol,n},\|{\cdot}\|_{\sol,n})\subseteq
\mathcal{B}((\cH_{\sol,n},\|{\cdot}\|_{\sol,n}),
(\cH_{\sol,n-l},\|{\cdot}\|_{\sol,n-l})),\quad t\mapsto L_t
\end{equation}
is $l$ times continuously differentiable with respect to
the strong operator topology on\\
$\mathcal{B}((\cH_{\sol,n},\|{\cdot}\|_{\sol,n}),
(\cH_{\sol,n-l},\|{\cdot}\|_{\sol,n-l}))$, where ``0 times continuously
differentiable'' means ``continuous''. 
In particular, for all $n\in\N_0$ the operator norm of $L_t$ is locally bounded
in $t\in\R$, i.e., $\sup_{t\in[t_0,t_1]}\Vert
L_t\Vert_{\cH_{\sol,n}\to\cH_{\sol,n}}<\infty$ for all reals $t_0<t_1$.
\end{lemma}

\begin{proof}
    Let $t\in\R$ and $n\in\N_0$ (not to be confused with the normal vector
    field $n=n(x)$ having the same name). 
    By applying Lemma~\ref{lem:multiplication-bounded} from the appendix
    to the function $Z(\vec x) = (v_t\slashed n_t\slashed A)(t_{\Sigma_t}(\vec
    x),\vec x)$, $\vec
    x\in\R^3$, it
    follows that the multiplication operator $v_t\slashed n_t\slashed
    A:(\cH_{\Sigma_t,n},\Vert{\cdot}\Vert_{\Sigma_t,n})\selfmaps$ is bounded.
    Consequently, $L_t = \F_{\sol\Sigma_t}(v_t\slashed n_t\slashed
    A)\F_{\Sigma_t\sol}:(\cH_{\sol,n},\Vert{\cdot}\Vert_{\sol,n})\selfmaps$ is
    bounded as well. \\

    Let $n,l\in\N_0$ such that $n\geq l$. It suffices to check that the
    maps $t\mapsto L_{t}\psi\in\cH_{0,n-l}$ have the regularity $\cC^l$
    w.r.t.\ the norm $\Vert{\cdot}\Vert_{0,n-l}$
    for all $\psi$ in the dense subset $\Cs$ of $\cH_{0,n}$.
    Furthermore, it even suffices to check the $\cC^l$ regularity of
    \begin{align}
        t\mapsto  \F_{\M\sol}L_t\psi 
        =
        \F_{\M\Sigma_t}(v_t\slashed n_t\slashed A)\F_{\Sigma_t,0}\psi
        \in
        (\cH_{\M,n-l},\Vert{\cdot}\Vert_{\M,n-l}).
    \end{align}
    Using Theorem~\ref{thm:F_maps} we find for $p\in\M$
    \begin{align}
        \F_{\M\sol}L_t\psi(p)
        =
        \frac{\slashed p+m}{2m}(2\pi)^{-3/2}
        \int_{\Sigma_t} e^{ipx}\,i_\gamma(d^4x)
        \,v_t(x)\slashed n_t(x)\slashed A(x)\psi(x)    
    \end{align}
    We recall that the motion of $\Sigma_t$ can be seen as driven by the vector field
    $v(x)\,n(x)$, that $A$ is compactly supported, and that
    $v(x),n(x),A(x),\psi(x)$ are smooth. Therefore, the
    following derivatives w.r.t.\ $t$, point-wise in $p\in\M$, exist and are given by
    \begin{align}
        \frac{\partial^l}{\partial t^l}
        \F_{\M\sol}L_t\psi(p)
        =
        \frac{\slashed p+m}{2m}(2\pi)^{-3/2}
        \int_{\Sigma_t} 
        \cal L_{v_t n_t}^l \left(
            e^{ipx}\,i_\gamma(d^4x)
            \,v_t(x)\slashed n_t(x)\slashed A(x)\psi(x)
        \right),
    \end{align}
    where $\cal L_{v_t n_t}=i_{v_t n_t}\circ d+d\circ i_{v_t n_t}$ 
    denotes the Lie derivative.
    Expanding the iterated Lie derivative, the integrand
    takes the form
    \begin{align}
        \cal L_{v_t n_t}^l \left(
            e^{ipx}\,i_\gamma(d^4x)
            \,v_t(x)\slashed n_t(x)\slashed A(x)\psi(x)
        \right)
        = \sum_{
            \substack{\alpha,\beta\in\N^4_0\\
            |\alpha+\beta|\leq l}}
        p^\beta e^{ipx} i_\gamma(d^4x) \zeta_{l,\alpha,\beta}(x) \partial^\alpha
        \psi(x)
    \end{align}
    for appropriate $\zeta_{l,\alpha,\beta} \in \cC_c^\infty(\R^4,\C^{4\times
    4})$ being independent of
    $p$ and $\psi$. It follows
    \begin{align}
        \frac{\partial^l}{\partial t^l}
        \F_{\M\sol}L_t\psi(p)
        =
        \sum_{
            \substack{\alpha,\beta\in\N^4_0\\
            |\alpha+\beta|\leq l}} p^\beta \F_{\M \Sigma_t}
        \left.\left(\zeta_{l,\alpha,\beta} \partial^\alpha
        \psi\right)\right|_{\Sigma_t}(p),
        \qquad\text{for }p\in\M.
    \end{align}
Using Lemma~\ref{lemma: FMSigma} we observe for our given $\psi\in\Cs$ 
that for all bounded intervals
$[t_0,t_1]$ there exists $\gamma>0$ such that for all
$k\in\N$
\begin{align}
    \sup_{\substack{t\in[t_0,t_1]\\\alpha,\beta\in\N^4_0,\,|\alpha+\beta|\leq l}}
    \left\Vert \F_{\M \Sigma_t}
   \left.\left(\zeta_{l,\alpha,\beta} \partial^\alpha
   \psi\right)\right|_{\Sigma_t} 
   \right\Vert_{\M,\gamma,k} 
   < \infty
\end{align}
holds. Using dominated convergence we infer that
\begin{align}
    [t_0,t_1]\ni t 
    \mapsto
    \sum_{
            \substack{\alpha,\beta\in\N^4_0\\
            |\alpha+\beta|\leq l}} p^\beta \F_{\M \Sigma_t}
        \left.\left(\zeta_{l,\alpha,\beta} \partial^\alpha
        \psi\right)\right|_{\Sigma_t}
    \in (\cH_{\M,j},\|{\cdot}\|_{\M,j})
    \label{eq:cont-c0}
\end{align}
is continuous for all $j\in\N_0$ and equals 
$t\mapsto \frac{\partial^l}{\partial t^l}
        \F_{\M\sol}L_t\psi$, where the derivatives are taken in
        $(\cH_{\M,j},\|{\cdot}\|_{\M,j})$.
By Lemma~\ref{lem:multiplication-bounded},
Definition~\ref{def:sobolev}, and the fact that $\cC_\sol$ is dense in
$\cH_{\sol,n}$, we conclude that the maps 
\begin{align}
    (\cH_{\sol,n},\|{\cdot}\|_{\sol,n}) \ni \phi
    \mapsto
    \sum_{
            \substack{\alpha,\beta\in\N^4_0\\
            |\alpha+\beta|\leq l}} p^\beta \F_{\M \Sigma_t}
        \left.\left(\zeta_{l,\alpha,\beta} \partial^\alpha
        \phi\right)\right|_{\Sigma_t}
        \in (\cH_{\M,n-l},\|{\cdot}\|_{\M,n-l})
\end{align}
are bounded uniformly in $t\in[t_0,t_1]$. Using this, the continuity of the map
in 
(\ref{eq:cont-c0}), and again
the denseness argument, we note that the
continuity claimed in (\ref{eq:cont-c0}) holds also for any
$\psi\in\cH_{\sol,n}$ and $j=n-l$.
By induction in $l=0,1,\ldots, n$ we find that
for any $\phi\in\cH_{\sol,n}$
\begin{align}
    \frac{\partial^l}{\partial t^l}
    \F_{\M\sol}L_t\phi
    =
    \sum_{
        \substack{\alpha,\beta\in\N^4_0\\
        |\alpha+\beta|\leq l}} p^\beta \F_{\M \Sigma_t}
    \left.\left(\zeta_{l,\alpha,\beta} \partial^\alpha
    \phi\right)\right|_{\Sigma_t},
\end{align}
where the derivative in the induction step is taken in
$(\cH_{\M,n-l},\|{\cdot}\|_{\M,n-l})$.
\end{proof}

\begin{lemma}[Existence and Uniqueness in the Interaction Picture]
  \label{lem:interaction-existence-uniqueness}
Let $n\in\N_0$.
%\begin{enumerate}
%    \item \todo{kommt raus}
%For any $t\in\R$, the operator $L_t:\Cs\selfmaps$
%introduced in (\ref{eq:interaction-evolution-Cs})
%extends to a bounded linear map 
%$L_t:(\cH_{\sol,n},\|{\cdot}\|_{\sol,n})\selfmaps$, 
%denoted by the same symbol $L_t$.
%It depends continuously on $t\in\R$ in the strong operator topology
%on $\mathcal{B}(\cH_{\sol,n},\|{\cdot}\|_{\sol,n})$.
%Its operator norm is bounded in $t\in\R$: 
%$\sup_{t\in\R}\Vert L_t\Vert_{\cH_{\sol,n}\to\cH_{\sol,n}}<\infty.$ 
%\item
For any $\chi\in \cH_{\sol,n}$, the initial value problem
over $(\cH_{\sol,n},\|{\cdot}\|_{\sol,n})$
\begin{align}
\label{eq:ivp}
i\frac{d}{dt}\phi_t=L_t\phi_t,\quad\phi_0=\chi
\end{align}
has a unique solution
$\phi_{(\cdot)}:\R\to\cH_{\sol,n}$ which is continuously
differentiable w.r.t.\ the norm $\|{\cdot}\|_{\sol,n}$.
%\end{enumerate}
\end{lemma}
Note that any solution of the initial value problem
(\ref{eq:ivp}) over $(\cH_{\sol,n},\|{\cdot}\|_{\sol,n})$
is also a solution over
$(\cH_{\sol,n'},\|{\cdot}\|_{\sol,n'})$ for any $n'\in\N_0$ with $n'\le n$.
In particular, for initial data $\chi\in\bigcap_{n\in\N_0}\cH_{\sol,n}$ the
corresponding solution $\phi_t$ lies in the same intersection of spaces.

%The next lemma concerns support properties of solutions.
%Its proof is given in Section~\ref{subsect:Regularity and support properties}.

%Next, we proof part on
%\begin{lemma}[Support of solutions I]
%\label{lemma: support}
%  Let $\chi_\Sigma\in\CSigma$ and let $\psi\in\CA(\R^4,\C^4)$ be a
%  solution to the Dirac equation (\ref{eq:dirac-equation}) with the initial
%  data $\psi|_\Sigma=\chi_\Sigma$. Then
%  $\supp\psi\subseteq\supp\chi_\Sigma+\causal$.
%\end{lemma}
%
%\begin{lemma}[Support of solutions II]
%\label{lemma: support}
%  Let $\chi_\Sigma\in\CSigma$ and let $\psi\in\cC^\infty(\R^4,\C^4)$ be a
%  solution to the Dirac equation (\ref{eq:dirac-equation}) with the initial
%  data $\psi|_\Sigma=\chi_\Sigma$. Then
%  $\supp\psi\subseteq\supp\chi_\Sigma+\causal$.
%\end{lemma}

%\begin{lemma}[Support of solutions]
%\label{lemma: support}
%Let $\chi\in\Cs$ be supported in $K+\causal$
%with a given compact set $K\subset\R^4$.
%Then the solution $\phi_\cdot$ of the initial value
%problem (\ref{eq:ivp}) fulfills
%\begin{align}
%\supp \phi_t\subseteq (K\cup \supp A)+\causal
%\end{align}
%for any $t\in\R$.
%\end{lemma}

\begin{proof}%[Proof of Lemma~\ref{lem:interaction-existence-uniqueness}]
Lemma \ref{lem:reg-Lt} ensures that the map $\R\ni t\mapsto L_t\in
\cal B(\cH_{\sol,n},\|{\cdot}\|_{\sol,n})$ is continuous with respect to the strong
operator topology. Consequently, by the Picard-Lindel\"of theorem, it follows that
the Volterra integral equation
associated to (\ref{eq:ivp})
\begin{equation}
\label{eq:Volterra}
\phi_t=\chi-i\int_0^t L_s\phi_s\,ds
\end{equation}
has a unique continuous 
solution $\phi_{(\cdot)}:\R\to(\cH_{\sol,n},\|{\cdot}\|_{\sol,n})$
for any initial value $\chi\in\cH_{\sol,n}$ and any $n\in\N_0$.
Furthermore, the fundamental theorem of calculus guarantees that 
it is continuously differentiable with respect to
the norm $\|{\cdot}\|_{\sol,n}$ with the derivative given in (\ref{eq:ivp}).
\end{proof}

\subsubsection{\db Regularity and Support Properties\ed}
\label{subsect:Regularity and support properties}
\begin{lemma}[Regularity of Solutions]
\label{lemma: regularity}
Let $n,l\in\N_0$ such that $n\geq l$. 
\begin{enumerate}
%    \item \todo{kommt raus}
%\label{claim: a}
%The function
%\begin{equation}
%    L_{(\cdot)}:\R\to\mathcal{B}(\cH_{\sol,n+l},\|{\cdot}\|_{\sol,n+l})\subseteq
%\mathcal{B}((\cH_{\sol,n+l},\|{\cdot}\|_{\sol,n+l}),
%(\cH_{\sol,n},\|{\cdot}\|_{\sol,n})),\quad t\mapsto L_t
%\end{equation}
%is $l$ times continuously differentiable with respect to
%the strong operator topology on\\
%$\mathcal{B}((\cH_{\sol,n+l},\|{\cdot}\|_{\sol,n+l}),
%(\cH_{\sol,n},\|{\cdot}\|_{\sol,n}))$.
\item
\label{claim: b}
For any initial value $\chi\in\cH_{\sol,n}$,
the solution $\R\ni t\mapsto \phi_t$ of the initial value problem
(\ref{eq:ivp}) is $l$ times continuously differentiable
w.r.t.\ the norm $\|{\cdot}\|_{\sol,n-l}$.
\item
\label{claim: c}
If in addition $n\ge l+2$, then the map
\begin{equation}
\phi:\R^4\times\R\to\C^4, \quad
\phi(x,t)=\phi_t(x)
\end{equation}
is well-defined and $l$ times continuously differentiable.
In particular, the function $\phi$ is smooth
for initial values $\chi\in\cC_\sol$ and solves the
initial value problem 
(\ref{eq:initial cond interaction picture}), 
(\ref{eq:interaction-evolution-Cs})
in the classical sense.
\end{enumerate}
\end{lemma}

In the proof of Lemma~\ref{lemma: regularity} we rely on the following lemma, which we prove first.

\begin{lemma}[Derivatives of Translation Maps and Pointwise Evaluation]
    \label{lemma: Ableitung translationen}
\mbox{}
\begin{enumerate}
\item
For any $n,l\in\N_0$ such that $n\geq l$, the family
of translation maps
\begin{align}
T^{-\cdot}_\sol:\R^4\to  \mathcal{B}((\cH_{\sol,n},\|{\cdot}\|_{\sol,n}),
(\cH_{\sol,n-l},\|{\cdot}\|_{\sol,n-l})),\quad y\mapsto T^{-y}_\sol
\end{align}
is $l$ times continuously differentiable
w.r.t.\ the strong operator topology
with the derivatives
\begin{equation}
\partial_y^\alpha T^{-y}_\sol=T^{-y}_\sol\partial^\alpha
\in \mathcal{B}((\cH_{\sol,n},\|{\cdot}\|_{\sol,n}),
(\cH_{\sol,n-l},\|{\cdot}\|_{\sol,n-l})),
\quad
(\partial_y^\alpha T^{-y}_\sol)\psi=\partial^\alpha(T^{-y}_\sol\psi)
\end{equation}
for every \db multi-index \ed $\alpha\in\N_0^4$ with $|\alpha|\le l$
and $\psi\in \cH_{\sol,n}$.
\item
Given $k,n\in\N$ with $n\ge k+2$,
let $\phi_{(\cdot)}:\R\to(\cH_{\sol,n},\|{\cdot}\|_{\sol,n})$, $t\mapsto\phi_t$,
be a $k$ times continuously differentiable map.
Then the function
$\phi:\R^4\times\R\to\C^4$, $\phi(x,t)=\phi_t(x)$, which is well-defined by
Lemma~\ref{lemma: Pointwise evaluation}, is $k$ times continuously
differentiable.
In particular, if $\phi_{(\cdot)}:\R\to\bigcap_{n\in\N_0}\cH_{\sol,n}$
is smooth w.r.t.\ all norms $\|{\cdot}\|_{\sol,n}$, $n\in\N_0$,
then the function
$\phi$ is also smooth.
\end{enumerate}
\end{lemma}
\begin{proof}
(a)  Using (\ref{eq: transl M}), we write $T^{-y}_\sol=\F_{\sol\M}T^{-y}_\M\F_{\M\sol}
=\F_{\sol\M}e^{-ipy}\F_{\M\sol}$ and
$\partial^\alpha=(-i)^{|\alpha|}\F_{\sol\M}p^\alpha\F_{\M\sol}$.
Because the operators $\F_{\sol\M}:
(\cH_{\M,n},\|{\cdot}\|_{\M,n})\to(\cH_{\sol,n},\|{\cdot}\|_{\sol,n})$
and
$\F_{\sol\M}:
(\cH_{\sol,n-l},\|{\cdot}\|_{\sol,n-l})\to(\cH_{\M,n-l},\|{\cdot}\|_{\M,n-l})$
are unitary, the claim is equivalent to showing that
\begin{align}
T^{-\cdot}_\M:\R^4\to  \mathcal{B}((\cH_{\M,n},\|{\cdot}\|_{\M,n}),
(\cH_{\M,n-l},\|{\cdot}\|_{\M,n-l})),\quad y\mapsto T^{-y}_\M
\end{align}
is $l$ times continuously differentiable
w.r.t.\ the strong operator topology
with the derivatives
\begin{equation}
\partial_y^\alpha T^{-y}_\M=(-i)^{|\alpha|}T^{-y}_\M p^\alpha.
\end{equation}
Convergence of the corresponding difference quotients
in the strong operator topology is a consequence of the dominated 
convergence theorem, and the claim follows.

(b) The map $\Phi:\R^4\times\R\to
(\cH_{\sol,k+2},\|{\cdot}\|_{\sol,k+2})$,
$\Phi(x,t)= T_\sol^{-x}\phi_t$ 
is $k$ times continuously partially differentiable w.r.t.\ the argument
$t$. Furthermore, viewing $\Phi$ as a map
$\Phi:\R^4\times\R\to
(\cH_{\sol,2},\|{\cdot}\|_{\sol,2})$,
part (a) implies that all
partial derivatives $\frac{\partial^l}{\partial t^l}\Phi(x,t)$,
$l=0,\ldots,k$, are $k$ times continuously differentiable w.r.t.\ to
the argument $x$. Finally, using Lemma~\ref{lemma: Pointwise evaluation}
and the fact
$\phi=\delta_0 \circ \Phi$
concludes the proof.
\end{proof} 

\begin{proof}[Proof of Lemma~\ref{lemma: regularity}]
%(a) Because $\Sigma$ is a {\it smooth} hypersurface and
%$A\in C^\infty_c(\R^4,\R^4)$, the
%map $Z:\Sigma\times\R\to\C^{4\times 4}$ defined in (\ref{def Z}) is smooth and
%compactly supported.
%Using Lemma~\ref{lem:multiplication-bounded}
%in the appendix, it follows that
%for any $n\in\N_0$, the family of multiplication operators
%$Z_{(\cdot)}:\R\to \mathcal{B}((\cH_{\Sigma,n},\|{\cdot}\|_{\Sigma,n})$
%is smooth w.r.t.\ the operator norm;
%hence, $\F_{\sol \Sigma} Z_{(\cdot)} \F_{\Sigma\sol}:
%\R\to \mathcal{B}((\cH_{\sol,n},\|{\cdot}\|_{\sol,n})$
%is also smooth w.r.t.\ the operator norm.
%Using the representation (\ref{eq: L_t}) of $L_t$ together with
%Lemma~\ref{lemma: Ableitung translationen}, claim (a)
%follows.
%
Claim (a) follows by induction over $l$,
using Lemma \ref{lem:reg-Lt} and taking derivatives
w.r.t.\ $t$ of the right hand side of
the differential equation~(\ref{eq:ivp}).
Claim (b) follows directly from 
part (a) using Lemma~\ref{lemma: Ableitung translationen} (b).
\end{proof}

\subsubsection{Proof of Theorem~\ref{lemma: equiv}}
\label{sect: class sol}
\begin{proof}[Proof of Theorem~\ref{lemma: equiv}]
The key to the claimed equivalence between the 
Schr\"odinger picture and the interaction picture
is contained in the following calculation:
Let $\phi:\R^4\times\R\to\C^4$ 
be a smooth function that solves the free Dirac
equation in the first
argument: $i\slashed\partial_x\phi(x,t) = m\phi(x,t)$.
Let $\psi:\R^4\to\C^{4}$ be the smooth function given by 
$\psi(x)=\phi(x,\tau(x))$.
Then we have the following equivalences for any $x\in\R^4$:
  \begin{align}
&    i\partial_t \phi(x,t) = v_t(x) \slashed n_t(x) \slashed A(x) \phi(x,t)
    \text{ at $t=\tau(x)$ }
    \cr
    \Leftrightarrow \qquad &
i\slashed n_t(x) v_t(x)^{-1}\partial_t \phi(x,t) = \slashed A(x) \phi(x,t)
\text{ at } t=\tau(x)
\cr
\qquad \Leftrightarrow \qquad & i\slashed\partial_y \phi(x,\tau(y))
    =
    i\slashed\partial_y \tau(y) \partial_t \phi(x,t) 
    =
    \slashed A(y) \phi(x,\tau(y))
    \text{ at } y=x, \; t=\tau(y)
    \cr
    \Leftrightarrow \qquad &
    \left( i\slashed\partial_y - \slashed A(y) - m \right) \phi(x,\tau(y)) = -m \phi(x,\tau(y))
    \text{ at } y=x, \; t=\tau(y)
    \cr
    \Leftrightarrow \qquad &
    \left( i\slashed\partial_x + i\slashed\partial_y - \slashed A(y) - m \right) 
    \phi(x,\tau(y)) = \left( i\slashed\partial_x - m\right) \phi(x,\tau(y)) = 0
    \text{ at } y=x, \; t=\tau(y)
    \cr
    \Leftrightarrow \qquad &
    \left( i\slashed\partial_x - \slashed A(x) - m \right) 
    \psi(x) 
    =
    \left( i\slashed\partial_x - \slashed A(x) - m \right) 
    \phi(x,\tau(x))
    = 0
\cr
    \Leftrightarrow \qquad &
\text{$\psi$ solves the Dirac equation (\ref{eq:dirac-equation}) at $x$
with potential $A$.} 
\label{eq: Aequivalenzkette} 
\end{align}
Here we used $\slashed n^2=1$,
Definition (\ref{eq:partial-tau}), and
$\left(i\slashed\partial_x-m\right) \phi(x,t)=0$ 
because $\phi_t\in\Cs$.\\

\db\textsc{First, we prove part ({\rm a}) of the theorem.} \ed
Let $\psi\in\CA$ and
$\phi:\R^4\times\R$ be as in the hypothesis of the theorem.
In particular, using smoothness and the support property
of $\psi$, the function
$\R\ni t\mapsto T^{-te_0}_{\Sigma_t}(\psi|_{\Sigma_t})$
takes values in $\CSigma\subseteq \cH_{\Sigma,n}$ 
for any $n\in\N_0$ and is smooth with respect to
the norm $\|{\cdot}\|_{\Sigma,n}$.
Because $\F_{\sol\Sigma}:(\cH_{\Sigma,n},\|{\cdot}\|_{\Sigma,n})\to
(\cH_{\sol,n},\|{\cdot}\|_{\sol,n})$ is unitary, the function
$\R\ni t\mapsto \F_{\sol\Sigma}T^{-te_0}_{\Sigma_t}\psi|_{\Sigma_t} 
=T^{-te_0}_\sol \db \F_{\sol\Sigma_t} \ed \psi|_{\Sigma_t} =T^{-te_0}_\sol\phi_t\in
 \Cs\subseteq 
\cH_{\sol,n}$ 
is smooth w.r.t.\ the norm  $\|{\cdot}\|_{\sol,n}$ as well.
\db As \ed this holds for all $n\in\N_0$,  
Lemma~\ref{lemma: Ableitung translationen}~(a) implies
that $t\mapsto \phi_t= T^{te_0}_\sol T^{-te_0}_\sol\phi_t\in 
\Cs\subseteq \cH_{\sol,n}$
is also smooth w.r.t.\ $\|{\cdot}\|_{\sol,n}$ for any $n\in\N_0$.
Now Lemma~\ref{lemma: Ableitung translationen}~(b)
shows that $\phi:\R^4\times\R\to\C^4$ is smooth.
By definition, $\phi(x,t)$ solves the free Dirac equation in 
$x$-argument, 
fulfills the initial condition
(\ref{eq:initial cond interaction picture}),
and $\phi(x,\tau(x))=\psi(x)$ holds for all $x\in\R^4$.
The direction ``$\Leftarrow$'' in the sequence
(\ref{eq: Aequivalenzkette}) of equivalences 
and the assumption $\psi\in\CA$ \db implies \ed
\begin{equation}
\label{eq: partial-t-phi-auf-Sigma-t}
i(\partial_t \phi_t)|_{\Sigma_t}=
v_t\slashed n_t \slashed A \db \F_{\Sigma_t\sol} \ed \phi_t.
\end{equation}
Next we  examine the support properties of all $\phi_t$
for $t$ in any given compact interval $I$.
The set $K_I:=\supp \psi\cap 
\bigcup_{t\in I}\Sigma_t$ is compact, and for any $t\in I$,
the function $\phi_t=\db \F_{\sol \Sigma_t}\ed \psi|_{\Sigma_t}$
is supported in $\supp(\psi|_{\Sigma_t})+\causal\subseteq K_I+\causal$ by Lemma
\ref{lemma: cycle Sigma-s-M}.
Because of $\phi_t\in\Cs$ for any $t$ and because
the Dirac operator $i\slashed\partial_x-m$ in the $x$-argument
commutes with the derivative $\partial_t$, this implies
$\partial_t\phi(\cdot,t)\in\Cs$ for any $t\in\R$.
Thus, we can rewrite (\ref{eq: partial-t-phi-auf-Sigma-t}) in the form
$
i\db \F_{\Sigma_t \sol} \ed \partial_t \phi_t=
v_t\slashed n_t \slashed A\F_{\Sigma_t\sol}\phi_t.
$
Because $\db \F_{\Sigma_t \sol} \ed :\Cs\to\cC_{\Sigma_t}$ and
$\db \F_{\sol\Sigma_t} \ed :\cC_{\Sigma_t}\to\Cs$ are inverse
to each other by Theorem~\ref{thm:F_maps},
we conclude that formula (\ref{eq:interaction-evolution-Cs}) holds:
\begin{equation}
i\partial_t \phi_t=\db \F_{\sol \Sigma_t}\ed
v_t\slashed n_t \slashed A\db \F_{\Sigma_t \sol}\ed\phi_t=L_t\phi_t.
\end{equation}
To finish the proof of part (a) of the theorem,
we examine the support of $\psi_t$ uniformly in $t\in\R$.
First, because the vector potential $A$ is compactly supported,
we have $L_t=0$ for all $t\in\R$ with $|t|$ large enough.
This shows that $\phi_t$ does not depend on $t$ for $t\ge t_0$ 
for some large enough $t_1>0$. 
The same holds for $t\le t_1$ for some $t_0<0$ small enough.
Using the compact interval
$I=[t_0,t_1]$ we reconsider the compact set $K:=K_I$ from above.
It follows $\supp\phi_t\subseteq K+\causal$ for all $t\in\R$,
not only for $t\in I$. 
Thus, part (a) of the lemma is proven.\\

\textsc{Next, we prove part }(b). By the assumption $\supp\phi\subseteq(K+\causal)\times\R$,
it follows $\supp\psi\subseteq K+\causal$.
Because $\phi$ fulfills
the evolution equation
(\ref{eq:interaction-evolution-Cs}),
the direction ``$\Rightarrow$'' in the sequence
(\ref{eq: Aequivalenzkette}) of equivalences implies
$\psi\in\CA$.
The remaining claims
$\psi|_\Sigma=\chi_\Sigma$ and
$\phi_t=\db \F_{0 \Sigma_t}\ed\psi|_{\Sigma_t}$ follow immediately from the
definitions.
\end{proof}

\subsubsection{Proof of Theorem~\ref{dirac-existence-uniqueness}}
\label{sec:main-result-1}

\begin{proof}[Proof of Theorem~\ref{dirac-existence-uniqueness}]
  We take a fixed future-directed foliation of space-time $\mathbf{\Sigma}$.
  Define $\chi:=\F_{0\Sigma}\chi_\Sigma$. By 
  Lemma~\ref{lem:interaction-existence-uniqueness} there is a solution
  $\phi_{(\cdot)}$ of the initial value problem (\ref{eq:ivp}).\\

  \textsc{First, we prove that $\supp\phi_t \cap \Sigma_t\subseteq \supp\chi_\Sigma + \causal$ for all
  $t\in\R$.}
  Let us assume $t\geq 0$.  
  We define $\causal_+:=\{x\in\causal|\;x^0\geq 0\}$ as well as
  $\Cs(\chi_\Sigma,t):=\{\eta\in\Cs
    \;|\; \supp\F_{\Sigma_t\sol}\eta \subseteq \supp \chi_\Sigma + \causal_+
  \}$ and its closure $\cH_{\sol,n}(\chi_\Sigma,t)$ in $(\cH_{\sol,n},
  \Vert{\cdot}\Vert_{\sol,n})$, $n\in\N_0$. Furthermore, for all $T\geq 0$ we
  define
  \begin{align}
    X_{T,n} := \left\{
        \varphi_{(\cdot)} \in \cC([0,T],(\cH_{\sol,n},\|{\cdot}\|_{\sol,n})) 
      \;|\; 
      \forall\;0\leq s\leq T:\; \varphi_s \in \cH_{\sol,n}(\chi_\Sigma,s)
    \right\}
  \end{align}
  which is a Banach space w.r.t.\ the norm
  $\|\varphi_{(\cdot)}\|_{X_{T,n}}:=\sup_{t\in[0,T]}\|\varphi_t\|_{0,n}$.
  As seen in the proof of Lemma~\ref{lem:interaction-existence-uniqueness}, 
  for any given $T\geq 0$ and $n\in\N_0$ 
  the Volterra integral
  equation (\ref{eq:Volterra}) gives rise to a self-map $S:
  \cC([0,T],(\cH_{\sol,n},\|{\cdot}\|_{\sol,n}))\selfmaps$
  \begin{align}
      S:\varphi_{(\cdot)} \mapsto \left(t\mapsto S_t(\varphi_{(\cdot)}):=
    \F_{\sol\Sigma}\chi_\Sigma-i\int_0^t
    L_s\varphi_s\,ds\right), \qquad \text{for }\varphi_{(\cdot)}\in
    \cC([0,T],(\cH_{\sol,n},\|{\cdot}\|_{\sol,n})).
  \end{align}
  We further claim that $S:X_{T,n}\selfmaps$. To see this we need to show
  $S_t(\varphi_{(\cdot)})\in\cH_{\sol,n}(\chi_\Sigma,t)$
  for all $\varphi_{(\cdot)}\in X_{T,n}$ and $0\leq t\leq T$. 
  First, thanks to Theorem~\ref{thm:F_maps} (b), the fact that
  $\mathbf{\Sigma}$ is future oriented, and $\causal_++\causal_+=\causal_+$,
  one has
  $\Cs(\chi_\Sigma,s)\subseteq
  \Cs(\chi_\Sigma,t)$
  for  all $0\leq s\leq t$. This implies
  that $\F_{\sol\Sigma}\chi_\Sigma\in \Cs(\chi_\Sigma,t)$ for all
  $0\leq t\leq T$. Second, with the help of the representation of $L_s$ given in
  (\ref{eq:interaction-evolution-Cs}) and using
  Theorem~\ref{thm:F_maps} (a), we observe that
  $L_s \eta\in \Cs(\chi_\Sigma,s)$ for any $\eta\in \Cs(\chi_\Sigma,s)$ and
  $0\leq s\leq T$. Using that $L_t:(\cH_{\sol,n},\Vert{\cdot}\Vert_{\sol,n})
  \selfmaps$ is bounded, which was proven in Lemma~\ref{lem:reg-Lt}, and that $\Cs(\chi_\Sigma,s)$ is
  dense in $\cH_{\sol,n}(\chi_\Sigma,s)$ w.r.t.\
  $\Vert{\cdot}\Vert_{\sol,n}$,
  we conclude $L_s\varphi_s\in \cH_{\sol,n}(\chi_\Sigma,s)\subseteq
  \cH_{\sol,n}(\chi_\Sigma,t)$
  for any
  $\varphi_{(\cdot)}\in\cH_{\sol,n}(\chi_\Sigma,s)$ and $0\leq s\leq t\leq T$.
  This proves $S:X_{T,n}\selfmaps$. 
 
  In consequence, the corresponding unique fixed-point $\phi_{(\cdot)}$ found in
  Lemma~\ref{lem:interaction-existence-uniqueness} fulfills $\phi_{(\cdot)}\in
  X_{T,n}$. In particular, all $\phi_t|_{\Sigma_t}$, $t\geq 0$, are supported in
  $(\supp\chi_\Sigma+\causal_+)\cap\Sigma_t$. This together with an 
  analogous argument for $t\leq 0$ implies
  $\supp\phi_t|_{\Sigma_t}\subseteq\supp\chi_\Sigma+\causal$ for $t\in\R$.\\

  \textsc{Part (}i\textsc{).}
  From part (b) of Lemma~\ref{lemma: regularity} we know that $\phi_{(\cdot)}$
  is smooth. Hence, part (b) of Theorem~\ref{lemma: equiv} implies that the
  function $\psi$ given by
  $\psi(x):=\phi(x,\tau(x))$ is in $\CA$ with $\psi|_\Sigma=\chi_\Sigma$.
  Furthermore, we have
  $\supp\psi\subseteq \supp\chi_\Sigma+\causal$.\\
  
  \textsc{Part (}ii\textsc{).}
  Let $\widetilde\psi\in\cC^\infty(\R^4,\C^4)$ solve
  the Dirac equation (\ref{eq:dirac-equation}) for
  $\widetilde\psi|_\Sigma=\chi_\Sigma$.
  Suppose $\widetilde\psi\in\CA$. Then, by part (a) 
  of Theorem~\ref{lemma: equiv} together with the uniqueness statement of
  Lemma~\ref{lem:interaction-existence-uniqueness} we get
  $\widetilde\psi=\psi$.
  Finally, we show that
  $\widetilde\psi\in\CA$.
 For this we use a duality argument. 
  For any
  Cauchy surface $\Sigma'$ and $\phi\in\CA$ the pairing
  $\left\langle\phi,\widetilde\psi\right\rangle_{\Sigma'} := \int_{\Sigma'} \overline{\phi(x)}
    i_\gamma(d^4x)\widetilde\psi(x)$, cf. (\ref{eq:scalar-product-sol}),
  is well-defined because the integrand is smooth and has compact support, just
  as $\phi|_{\Sigma'}$. Recalling (\ref{eq:stokes-argument}),
  we find 
  $d[\overline{\phi(x)}i_\gamma(d^4x)\widetilde\psi(x)]=0$ so that for all Cauchy
surfaces $\Sigma'$ we have
$\langle\phi,\widetilde\psi\rangle_\Sigma=\langle\phi,\widetilde\psi\rangle_{\Sigma'}$. 
For any $\varphi_{\Sigma'}\in\cC_{\Sigma'}$ with $\supp\varphi_{\Sigma'} \cap
(\supp\chi_\Sigma + \causal) = \emptyset$ one has $\left( \supp\varphi_{\Sigma'}
+ \causal \right) \cap
\supp\chi_\Sigma = \emptyset$.
Consider $\phi\in\CA$ with $\phi|_{\Sigma'}=\varphi_{\Sigma'}$ the existence of
which is ensured by part (i). Then by part (i)
we know that $\supp \phi \subseteq \left( \supp \varphi_{\Sigma'} + \causal
\right)$, and hence, $\supp \phi \cap \supp \chi_\Sigma = \emptyset$.
We \db conclude \ed
\begin{align}
    \int_{\Sigma'} \overline{\varphi_{\Sigma'}(x)}
    i_\gamma(d^4x)\widetilde\psi(x)
    =\langle\phi,\widetilde\psi\rangle_{\Sigma'}
    =\langle\phi,\widetilde\psi\rangle_{\Sigma}
    =\int_{\Sigma} \overline{\phi(x)} i_\gamma(d^4x)\chi_\Sigma(x)=0
\end{align}
  as $\widetilde\psi|_\Sigma=\chi_\Sigma$. Since $\widetilde\psi$ is continuous we conclude that
  $\supp \widetilde\psi \subseteq \supp\chi_\Sigma + \causal$.
  Therefore, $\widetilde\psi\in\CA$.
\end{proof}

\subsubsection{\db Proof of Theorem~\ref{thm: Hilbert space time evolution}\ed}
\label{sec:evolution-hilbert}
\begin{proof}[Proof of Theorem~\ref{thm: Hilbert space time evolution}]
Theorem~\ref{dirac-existence-uniqueness} implies:
For any Cauchy surface $\Sigma$ and any vector potential
$A\in C^\infty_c(\R^4,\R^4)$, the restriction map
$\F_{\Sigma A}:\CA\to\CSigma$,  $\F_{\Sigma A}\psi=\psi|_{\Sigma}$
is a bijection. Let $\F_{A\Sigma}:\CSigma\to\CA$ denote its inverse.
Moreover, by the definition of the scalar product
on $\CA$ given in (\ref{eq:scalar-product-sol}), see also 
the argument~(\ref{eq:stokes-argument}), the restriction map
$\F_{\Sigma A}:\CA\to\CSigma$ is isometric. 
Taking the closure of this map, it has a unitary extension
$\F_{\Sigma A}:\cH_A\to\cH_\Sigma$ with a unitary inverse 
$\F_{A\Sigma}:\cH_\Sigma\to\cH_A$. The operator
$\F^A_{\Sigma'\Sigma}:=\F_{\Sigma' A}\F_{A\Sigma}:
\cH_\Sigma\to\cH_{\Sigma'}$ is then the unique unitary extension
of the isometric bijection
$\F^A_{\Sigma'\Sigma}=\F_{\Sigma' A}\F_{A\Sigma}:
\cC_\Sigma\to\cC_{\Sigma'}$.
 \end{proof}

\appendix

\section{Auxiliary Results}
\label{sec:appendix}
In this appendix, we prove some of the technical \db lemmas \ed used in the
rest of the paper. The first lemma deals with the inequalities
\db controlling \ed the geometry of the complexified mass shell.
\begin{lemma}[Geometric Properties of $\M_\C$]
\label{lemma: mass shell geometry}
For $p=(p^0,\vec p)\in\M_\C$,
one has the inequalities
\begin{align}
\label{p0 bound 2}
&|\im p^0|\le |\im \vec p|,
\\&
\label{p0 bound 3}
|\im \vec p|\le|\im p|\le \sqrt{2}|\im \vec p|, 
\\&
\label{p0 bound 5}
m\le |p|\le \sqrt{3}(m\vee|\vec p|),
\\&
\label{p0 bound 6}
m\vee
|\vec p|\le \constl{aux2}(\tfrac{m}{12}\vee|p^0|) e^{\epsilon |\im \vec  p|}
\quad\text{for } \epsilon>0,
\end{align}
with a constant $\constr{aux2}=\constr{aux2}(\epsilon m)>0$.
\end{lemma}
\begin{proof}
    For the first claim (\ref{p0 bound 2}), we calculate
with the notation $\vec p^2=(p^1)^2+(p^2)^2+(p^3)^2$:
\begin{align}
&2|\im p^0|^2=|p^0|^2-\re((p^0)^2)
=|\vec p^2+m^2|-\re(\vec p^2+m^2)
\nonumber\\&
\le
|\vec p^2|+m^2-\re(\vec p^2)-m^2
=|\vec p^2|-\re(\vec p^2)
\le |\vec p|^2-\re(\vec p^2)=2|\im \vec p|^2.
\end{align}
This also yields claim (\ref{p0 bound 3}):
\begin{align}
|\im \vec p|^2\le|\im p|^2=|\im p^0|^2+|\im \vec p|^2\le 2|\im \vec p|^2.
\end{align}
The inequality on the right in Claim (\ref{p0 bound 5}) 
follows directly from 
\begin{align}
|p|^2=|p^0|^2+|\vec p|^2=|\vec p^2+m^2|+|\vec p|^2\le 2|\vec p|^2+m^2
\le 3(|\vec p|^2\vee m^2).
\end{align}
The bound $m\le |p|$ is a consequence of $m^2=|p^2|\le|p|^2$.
To \db finally \ed prove  (\ref{p0 bound 6}), 
we observe that $\vec p^2+m^2=(p^0)^2$ implies
\begin{align}
|\vec p|^2-2|\im \vec p|^2+m^2=\re(\vec p^2)+m^2=\re((p^0)^2)\le |p^0|^2,
\end{align}
and hence,
$|\vec p|^2+m^2\le|p^0|^2+2|\im \vec p|^2$.
We get $m^2\vee |\vec p|^2\le |\vec p|^2+m^2\le 3(|p^0|^2\vee |\im \vec p|^2)$
and thus
\begin{equation}
m\vee |\vec p|\le \sqrt{3}(|p^0|\vee|\im \vec  p|)
\le  \constr{aux2}(\tfrac{m}{12}\vee|p^0|)(1+\epsilon |\im \vec  p|)
\le \constr{aux2}(\tfrac{m}{12}\vee|p^0|) e^{\epsilon |\im \vec  p|}
\end{equation}
with a constant
$\constr{aux2}=\constr{aux2}(\epsilon m)>0$.
\end{proof}

\begin{proof}[Proof of Lemma~\ref {lemma: p0 nahe null}]
Heuristically, the idea is to choose $k(t)=(k^0(t),\vec k(t))$ 
by a Euler substitution
as a rational function of $t$ such that $\vec k'(0)$ is small
whenever $p^0$ is close to $0$. More precisely, we proceed as follows.
We abbreviate $\delta:=1/12$, $\epsilon:=1/3$, and
$\gamma=1/6$.
The only facts that we need about these positive constants
are the following:
\begin{align}
\label{def H}
H:=2(\sqrt{1-\delta^2}-\epsilon-\delta)>0,
\\
\label{sqrt delta epsilon}
\gamma\ge \frac{\epsilon(\epsilon+2\delta)}{H},
\\
\label{eps gamma delta}
\epsilon-2\delta-\gamma\ge 0.
\end{align}
Given  $p=(p^0,\vec p)\in\M_{\C}$ with 
\begin{equation}
\label{p0delta}
|p^0|\le \delta m,
\end{equation}
we observe 
\begin{equation}
\label{lower bound vec p} 
|\vec p|\ge m\sqrt{1-\delta^2}>0
\end{equation}
from 
\begin{equation}
|\vec{p}|^2\ge |\vec p^2|=|(p^0)^2-m^2|
\ge m^2-|p^0|^2\ge m^2(1-\delta^2).
\end{equation}
We set 
$\vec q=(q^1,q^2,q^3):=\vec p^*/|\vec p|$,
where $\vec p^*$ 
denotes the complex conjugate of $\vec p$, and take $r\in\C$ with
$r^2=\vec q^2$.
It fulfills 
$|r|\le 1$
because of $|r|^2=|\vec{q}^2|\le|\vec q|^2=1$.
We set for $t\in \C$:
\begin{equation}
h(t):=2(\vec p\vec q-\epsilon mrt-p^0r)=2(|\vec p|-\epsilon mrt-p^0r),
\end{equation}
where we have abbreviated $\vec p\vec q:=\sum_{j=1}^3 p^j q^j$.
Using  (\ref{lower bound vec p}), $|r|\le 1$,
(\ref{p0delta}), and (\ref{def H}),
we get the following for $t\in\bar\Delta$.
\begin{equation}
|h(t)|\ge 2(|\vec p|-\epsilon m-|p^0|)\ge mH>0
\end{equation}
In particular,  
\begin{equation}
g:\bar\Delta\to\C,\quad
g(t):=\frac{(\epsilon m t)^2+2p^0\epsilon m t}{h(t)}
\end{equation}
is well-defined and extends to a 
holomorphic function in a neighborhood of $\bar \Delta$.
Using (\ref{sqrt delta epsilon}), it
fulfills the following bound for $t\in\bar \Delta$.
\begin{equation}
\label{bound g}
|g(t)|\le m \frac{\epsilon(\epsilon+2\delta)}{H}\le \gamma m
\end{equation}
We introduce $k(t)=(k^0(t),\vec k(t))$ for $t\in\bar\Delta$
by
\begin{align}
k^0(t)&:=p^0+\epsilon m t+g(t)r,
\\
\vec k(t)&:=\vec p+g(t)\vec q
\end{align}
 Note that $k$ \db also \ed extends to a holomorphic function in a neighborhood of
$\bar\Delta$.  Expanding the squares and using $r^2=\vec q^2$
and $(p^0)^2-\vec p^2=m^2$
\db we observe that $k(t)\in\M_{\C}$ as \ed
\begin{align}
k^0(t)^2-\vec k(t)^2= (p^0)^2-\vec p^2+
g(t)^2r^2-g(t)^2\vec q^2
+(\epsilon m t)^2+2p^0\epsilon m t
-g(t)h(t)=m^2.
\end{align}
The claim $k(0)=p$ follows from $g(0)=0$.
To obtain  claim (\ref{claim ktp}), we estimate for $t\in\bar\Delta$:
\begin{align}
|\vec k(t)-\vec p|=|g(t)\vec q| =|g(t)|\le\gamma m.
\end{align}
Finally, claim (\ref{claim k0p}) for
$t\in\C$ with $|t|=1$ is obtained from (\ref{p0delta}),
(\ref{bound g}), $|r|\le 1$, and (\ref{eps gamma delta})
as follows. 
\begin{align}
|k^0(t)|\ge|\epsilon m t|-|p^0|-|g(t)r|
\ge \epsilon m-\delta m-\gamma m\ge \delta m.
\end{align}
\end{proof}

\begin{proof}[Proof of Lemma~\ref{lemma: Poincar'e transformations}]
It is obvious that $T_\Sigma^{-y}$ maps
$\CSigma$ to $\cC_{\Sigma-y}$, $T_{A}^{-y}$ maps
$\CA$ to $\cC_{A(\cdot+y)}$, and
$L_\Sigma^{(S,\Lambda)}$ maps $\CSigma$ to $\cC_{\Lambda\Sigma}$.

To see that $T_{\mathcal M}^{-y}$ maps $\CM$ to $\CM$,
we consider $\psi\in \CM$, its holomorphic extension $\Psi:\M_\C\to \C^4$
and $\alpha>0$ with $\|\psi\|_{\M,\alpha, n}<\infty$ for all $n\in\N$.
Using inequality (\ref{p0 bound 3}), 
we conclude for any $n\in \N$
\begin{align}
\|T_{\mathcal M}^{-y}\psi\|_{\M,\alpha+\sqrt{2}|y|, n}
=
\sup_{p\in\M_{\C}} |p|^{n-1}e^{-(\alpha+\sqrt{2}|y|) |\im \vec p|}
e^{y\im p}|\Psi(p)|
\le 
\|\psi\|_{\M,\alpha, n}<\infty.
\end{align}
This proves $T_{\mathcal M}^{-y}\psi\in\CM$.

To show that $L_{A}^{(S,\Lambda)}$ maps $\CA$ to $\cC_{
    \Lambda A(\Lambda^{-1}\cdot)}$,
we take any $\psi\in\CA$ and a compact set $K\subset\R^4$
with $\supp\psi\subseteq K+\causal$.
Because of $\Lambda\causal=\causal$, we infer
$\supp(L_{A}^{(S,\Lambda)}\psi)\subseteq \Lambda K+\causal$,
and $\Lambda K$ is compact. Furthermore, 
$\psi':=L_{A}^{(S,\Lambda)}\psi$ is smooth
and solves the Dirac equation subject to the transformed external potential.
To see this, using the notation $x'=\Lambda x$, $\partial'_\mu=\partial/\partial x'^\mu$,
and $\partial_\nu=\partial/\partial x^\nu$, we note
\begin{align}
        m\psi'(x')
    =
        S m \psi(x)
    =
        S \gamma^\nu
        \left(
            i\partial_\nu - A_\nu(\Lambda^{-1}x')
        \right) 
        \psi(x).
    \label{eq:cursed-identity}
\end{align}
Formula (\ref{relation Lambda S}) and the identity ${\delta^\mu}_\nu =
{\Lambda_\sigma}^\mu {\Lambda^\sigma}_\nu$ implies
    $S \gamma^\nu
    =
    {\Lambda_\mu}^\nu \gamma^\mu S,$
and therefore
\begin{align}
    (\ref{eq:cursed-identity})
    & =
        {\Lambda_\mu}^\nu \gamma^\mu S
        \left(
            i\partial_\nu - A_\nu(\Lambda^{-1}x')
        \right) 
        \psi(x)
    \\
    & =
        \gamma^\mu
        \left(
            i\partial_\mu' - {\Lambda_\mu}^\nu A_\nu(\Lambda^{-1}x')
        \right) 
        \psi'(x'),
\end{align}
where we have use $\partial_\mu' = {\Lambda_\mu}^\nu \partial_\nu$.
This shows $\psi'\in\cC_{\Lambda A(\Lambda^{-1}\cdot)}$.

Finally, to see that 
$L_{\mathcal M}^{(S,\Lambda)}$ maps $\CM$ to $\CM$, 
we take $\phi\in\CM$, its holomorphic extension
$\Phi:\M_\C\to\C^4$, any vector $p'=(p'^0,\vec p')\in\M_\C$, 
$\alpha>0$ with $\|\psi\|_{\M,\alpha, n}<\infty$
for all $n\in\N$,
and set 
$\alpha'=\sqrt{2}\|\Lambda^{-1}\|\alpha$ and 
$p=(p^0,\vec p)=\Lambda^{-1}p'\in\M_\C$.
We get the following with inequality (\ref{p0 bound 3}): 
$\alpha' |\im \vec p'|\ge \alpha' |\im p'|/\sqrt{2}
\ge \alpha' \|\Lambda^{-1}\|^{-1}|\im p|/\sqrt{2}
\ge \alpha|\im \vec p|$,
and hence, for all $n\in\N$:
\begin{align}
|p'|^{n-1}e^{-\alpha'|\im \vec p'|}|S\Phi(\Lambda^{-1}p')|
\le \|S\|\|\Lambda\|^{n-1}|p|^{n-1}
e^{-\alpha|\im p|}|\Phi(p)|\le
\|S\|\|\Lambda\|^{n-1}\|\phi\|_{\M,\alpha, n},
\end{align}
using Definition~(\ref{bound 2a}).
This proves $\|L_{\mathcal M}^{(S,\Lambda)}\phi\|_{\M,\alpha', n}<\infty$
and therefore $L_{\mathcal M}^{(S,\Lambda)}\phi\in\CM$.
It is obvious that the six maps
(\ref{eq: transl Sigma})--
(\ref{eq: Lorentz M})
are invertible with inverses 
$T_\Sigma^y$, 
$T_{A}^y$,
$T_{\mathcal M}^y$,
$L_\Sigma^{(S^{-1},\Lambda^{-1})}$,
$L_{A}^{(S^{-1},\Lambda^{-1})}$, and
$L_{\mathcal M}^{(S^{-1},\Lambda^{-1})}$, respectively.
Furthermore, they are isometric. This is obvious in the case of the three translation maps
(\ref{eq: transl Sigma})--
(\ref{eq: transl M}).
We consider now Lorentz transformations 
(\ref{eq: Lorentz Sigma})--(\ref{eq: Lorentz M}). 
For $\phi,\psi\in\CSigma$, $\phi'=L_\Sigma^{(S^{-1},\Lambda^{-1})}\phi$,
and  $\psi'=L_\Sigma^{(S^{-1},\Lambda^{-1})}\psi$
we get by (\ref{eq:scalar-product}),
(\ref{lorentz1}), and
the invariance relation (\ref{eq:invarianz-i_gamma}):
\begin{align}
\sk{\phi',\psi'}&
=\int_{\Lambda\Sigma}\overline{\phi'(x')}\,i_\gamma(d^4x')\,\psi'(x')
=
\int_{\Lambda\Sigma}\overline{S\phi(\Lambda^{-1}x')}
\,i_\gamma(d^4x')\,S\psi(\Lambda^{-1}x')
\\&
=
\int_{\Lambda\Sigma}\overline{\phi(\Lambda^{-1}x')}
\gamma^0 S^*\gamma^0 \,i_\gamma(d^4x')\,S\psi(\Lambda^{-1}x')
=
\int_{\Lambda\Sigma}\overline{\phi(\Lambda^{-1}x')}
S^{-1} \,i_\gamma(d^4x')\,S\psi(\Lambda^{-1}x')
\\&
=
\int_{\Sigma}\overline{\phi(x)}
\,i_\gamma(d^4x)\,\psi(x)
=\sk{\phi,\psi}.
\end{align}
Using (\ref{eq:scalar-product-sol}), the same calculation is valid for
$\phi,\psi\in\CA$, $\phi'=L_{A}^{(S^{-1},\Lambda^{-1})}\phi$,
and  $\psi'=L_{A}^{(S^{-1},\Lambda^{-1})}\psi$.

For $\phi,\psi\in\CM$, $\phi'=L_\M^{(S^{-1},\Lambda^{-1})}\phi$,
and  $\psi'=L_\M^{(S^{-1},\Lambda^{-1})}\psi$,
the fact $m_\Lambda\M=\M$, equations
(\ref{eq:sk_M}), (\ref{lorentz1}), and the invariance relation
(\ref{eq: invarianz-ipd4p})
yield
\begin{align}
&\sk{\phi',\psi'}
=
\int_\M\overline{S\phi(\Lambda^{-1}p')}
S\psi(\Lambda^{-1} p')\,\frac{i_{p'}(d^4 p')}{m}
=
\int_\M\overline{\phi(\Lambda^{-1}p')}
\psi(\Lambda^{-1} p')\,\frac{i_{p'}(d^4 p')}{m}
\\&
=
\int_\M\overline{\phi(p)}
\psi(p)\,\frac{i_p(d^4 p)}{m}
=
\sk{\phi,\psi}
.
\end{align}
Since the six maps (\ref{eq: transl Sigma})-(\ref{eq: Lorentz M}) are isometric bijections, it
follows that they extend to unitary maps on the respective Hilbert spaces.
\end{proof}

\begin{proof}[Proof of Lemma~\ref{lemma:gauge}]
    First, we show that for a given $\lambda\in\cal C^\infty_c(\R^4,\R)$ the
    multiplication operator $\Gamma_\lambda$ maps $\CA$ to
    $\cC_{A+\partial\lambda}$.
    To show this, we take a $\psi\in\CA$ and define $\psi'(x):=\Gamma_\lambda
    \psi(x)=e^{-i\lambda(x)}\psi(x)$ for $x\in\R^4$. Clearly, $\supp\psi'=\supp\psi$ and $\psi'\in\cal
    C^\infty(\R^4,\C^4)$. The wave function $\psi'$ fulfills
    the Dirac equation subject to the transformed potential
    $A':=A+\partial\lambda$:
    \begin{align}
        \left(
            i\slashed \partial - \slashed A'(x)
        \right)
        \psi'(x)
        &=
        \left(
            i\slashed \partial- \slashed A'(x)
        \right)
        e^{-i\lambda(x)}\psi(x)
        \\
        &=
        e^{-i\lambda(x)}(i\slashed \partial - \slashed A(x))\psi(x)
        \\
        &=
        e^{-i\lambda(x)}m\psi(x) = m \psi'(x).
    \end{align}
    It is obvious that the map (\ref{eq:gauge}) is invertible and isometric.
    Therefore, it extends uniquely to a unitary map
    $\Gamma_\lambda:\HA\to\cH_{A+\partial\lambda}$.
\end{proof}

\begin{proof}[Proof of Theorem~\ref{thm: compatibility Poincar'e}]
Equations (\ref{eq:translation-M-Sigma}) and (\ref{eq:translation-s-M})
are obvious.
To prove (\ref{eq:Lorentz-M-Sigma}),
let $\psi\in\CSigma$, $\psi'=L_\Sigma^{(S,\Lambda)}\psi$,
and $p'=\Lambda p\in\M$.
Using the invariance relation (\ref{eq:invarianz-i_gamma})
and the consequence $\slashed p'=S\slashed pS^{-1}$
of (\ref{lorentz1}) and (\ref{relation Lambda S}), we obtain
\begin{align}
[\F_{\M, \Lambda\Sigma}L_\Sigma^{(S,\Lambda)}\psi](p')
&=\frac{\slashed p'+m}{2m}(2\pi)^{-3/2}\int_{\Lambda\Sigma} 
e^{ip'x'}\,i_\gamma(d^4x')\,\psi'(x')
\\&
=
S\frac{\slashed p+m}{2m}(2\pi)^{-3/2}\int_{\Lambda\Sigma} 
e^{ip'x'}S^{-1}\,i_\gamma(d^4x')\,S\psi(\Lambda^{-1}x')
\\&
=
S\frac{\slashed p+m}{2m}(2\pi)^{-3/2}\int_{\Sigma} 
e^{ipx}\,i_\gamma(d^4x)\psi(x)
\\&
=
[L_{\mathcal M}^{(S,\Lambda)}\F_{\M\Sigma}\psi](p').
\end{align}
Finally, equation~(\ref{eq:Lorentz-s-M})
is an immediate consequence of the invariance relation
(\ref{eq: invarianz-ipd4p}).
\end{proof}

The next auxiliary lemma deals with multiplication operators
in the Sobolev space $\cH_{\Sigma,n}$.
Let $\Sigma$ be a Cauchy surface and $K$ an open and relatively compact
subset of $\R^3$. Given $n\in\N_0$ we endow
$\cC^n_c(K,\C^{4\times4})$ with the norm
\begin{align}
  \Vert Z \Vert_{K,n,\infty} 
  :=
  \sum_{\substack{\beta\in\N_0^3\\|\beta|\leq n}}
  \sup_{\vec x\in K} 
  \left| 
    D^\beta Z(t_\Sigma(\vec x),\vec x) 
  \right|,
\quad \text{ where }\quad D^\beta=D_1^{\beta_1}D_2^{\beta_2}D_3^{\beta_3}.
\end{align} 

\begin{lemma}
   \label{lem:multiplication-bounded}
   There is a well-defined bounded linear map
   \begin{align}
     (\cC^n_c(K,\C^{4\times4}),\|{\cdot}\|_{K,n,\infty}) 
     \to 
     ({\cal B}(\cH_{\Sigma,n},\Vert{\cdot}\Vert_{\Sigma,n}),
     \Vert{\cdot}\Vert_{\cH_{\Sigma,n}\to\cH_{\Sigma,n}}),
     \qquad
     Z \mapsto (\psi \mapsto Z\psi).
   \end{align}
\end{lemma}

\begin{proof}
  For any $j=0,1,2,3$ and using (\ref{eq:partial-k}) and (\ref{eq:D-k}) we compute
  \begin{align}
    &i \partial_j Z \psi 
    =
    \sum_{k=1}^3 \alpha^\Sigma_{jk} D_k ( Z\psi )
    + \beta^\Sigma_j Z\psi
    \cr
    &= \sum_{k=1}^3 \alpha^\Sigma_{jk}Z \partial_k\psi+
    \sum_{k=1}^3 \alpha^\Sigma_{jk}\frac{\partial t_\Sigma}{\partial x^k} Z
       \partial_0\psi+
\left[\sum_{k=1}^3 \alpha^\Sigma_{jk} (D_kZ)+ \beta^\Sigma_jZ\right]\psi
    =: \sum_{\substack{\alpha\in\N_0^4\\|\alpha|\le 1}} 
     M_{\beta,\alpha}\partial^\alpha\psi
\end{align}
with $\beta=(\delta_{ij})_{i=0,1,2,3}$.
Iterating this formula for a general multi-index $\beta$, 
with $Z$ replaced by $M_{\beta,\alpha}$ and with
$\psi$ replaced by $\partial^\alpha\psi$ in the induction step,
yields
\begin{align}
i^{|\beta|}\partial^\beta Z \psi =\sum_{\substack{\alpha\in\N_0^4\\|\alpha|\le n}} 
     M_{\beta,\alpha}\partial^\alpha\psi,\quad
\quad
\|\partial^\beta Z \psi\| \le\sum_{\substack{\alpha\in\N_0^4\\|\alpha|\le n}} 
     \|M_{\beta,\alpha}\|_\infty\|\partial^\alpha\psi\|,\quad
\end{align}
 for $\beta\in\N_0^4$, $|\beta|\le n$,
with some $\C^{4\times 4}$-valued continuous functions $M_{\beta,\alpha}$ 
compactly supported in $K$ and depending
linearly on $D^\gamma Z$, $|\gamma|\le n$. Taking
the square and summing over $|\beta|\le n$  yields the claim.
\end{proof}

%--------------------------------------

\end{document}